\documentclass[11pt,letterpaper]{article}
\usepackage[left=1in, right=1in, top=1in, bottom=1in]{geometry}
\usepackage{algorithm}
\usepackage{algorithmic}
\usepackage{local}

\title{Ex-ante versus Ex-post: \\
Egalitarian Facility Location Mechanism Design}
\author{Zohar Barak\thanks{Tel Aviv University, \url{zoharbarak@mail.tau.ac.il}}
\and
Inbal Talgam-Cohen\thanks{Tel Aviv University, \url{inbaltalgam@gmail.com}}
}
\date{}

\begin{document}

\maketitle
\pagenumbering{gobble}

\begin{abstract}
We study the facility location mechanism design problem where the input is $n$ locations in Euclidean space reported by strategic agents, and the output is a single facility location. The cost of each agent is the distance from the returned facility, and our objective is to minimize the \emph{egalitarian} cost function (which is the maximum agent cost) in a strategyproof way.

It is known that the optimal deterministic approximation ratio is $2$, achieved by any deterministic dictator mechanism.
In this work, we study the power of randomized strategyproof-in-expectation mechanisms for this problem.
The literature has focused so far on \emph{ex-post} evaluation of such mechanisms, defined as the \emph{expected maximum} agent cost. In this work we turn to \emph{ex-ante} evaluation, defined as the \emph{maximum expected} agent cost, which is well-aligned with strategyproof-in-expectation, defined as \emph{maximizing the expected} agent utility through truth-telling.
We establish the following results:

\begin{enumerate}
    \item \textbf{Low dimensional space: Strict ex-ante vs.~ex-post separation.} In $\R$ and $\R^2$, 
    we show a strict separation between the ex-ante and ex-post objectives: For $\R$ we show a simple strategyproof mechanism that optimizes the ex-ante objective, while the best randomized mechanism for the ex-post objective achieves an approximation ratio of $1.5$. 
    For $\R^2$ we design the ``Random Rotated Corner" mechanism for the ex-ante objective, with an approximation ratio of at most $1.598$, hence breaking the deterministic barrier for this objective. For the ex-post objective, we give a lower bound of $1.605$, showing a strict separation between the two objectives in $\R^2$.

    \item \textbf{High dimensional space: Impossibility.} In high dimensional space $\R^d$ ($d \gg 1$), we show that no strategyproof in expectation mechanism is better than the deterministic dictator mechanism up to $o_d(1)$.
    This implies that neither ex-post nor ex-ante lead to improved fairness guarantees 
    in high dimensional space.
    One interesting implication is that the currently best known mechanism for the \emph{utilitarian} objective (minimizing the sum of agent costs), ``Random Rotation Coordinate-Wise Median" ($\rrcwm$), is simultaneously best possible for the egalitarian objective in high dimension, achieving an approximation ratio of $2$ for ex-post (and ex-ante) in $\R^d$ for any $d \ge 1$.
\end{enumerate}

\end{abstract}
\newpage
\tableofcontents
\newpage
\pagenumbering{arabic}
\setcounter{page}{1}

\section{Introduction}

Facility location mechanism design is a fundamental problem in mechanism design without money, in which $n$ agents report their locations in space, and a mechanism selects a facility location that optimizes a given objective while remaining strategyproof. Each agent's \emph{cost} is their distance from the selected location.

This problem was studied in the field of social choice~\cite{moulin1980strategy,border1983straightforward,kim1984nonmanipulability,peters1993range,barbera1993generalized,ching1997strategy,peremans1997strategy,barbera1998strategy,schummer2002strategy}, where it has been the primary domain with structured preferences that allows one to escape strong impossibility results such as the Gibbard-Satterthwaite Theorem~\cite{Gibbard77,Satterthwaite75}. The computer science community has been concerned with establishing approximate guarantees while maintaining strategyproofness, and extending the original strategic facility location setting in multiple directions (see, e.g., \cite{procaccia2013approximate, FotakisT14, FotakisT16, SerafinoV16, walsh2020strategy, agrawal2022learning, gravin2025approximation, ijcai2025p32,chan2026strategyproofmechanismseuclideanfacility,hastings2026strategicfacilitylocationpnorm}, and see~\cite{ijcai2021p596} for a survey on the topic). The setting often serves as an important domain for testing new concepts in approximate mechanism design, such as truthful mechanisms without monetary transfers~\cite{procaccia2013approximate}. Randomized mechanisms where studied by~\cite{AlonFPT10, LuSWZ10,FotakisT10,FeldmanW13, balkanski2024randomized, barak2024mac, ijcai2025p32, barak2026facilitylocationmechanismdesign}, among others.

The question is already interesting in low dimensional spaces like $\R$ and ${\R}^2$, which capture natural settings like picking a geographical location, or settling a political debate where the population is divided along one or two axes. While the one-dimensional case ($\R$) is relatively well-understood and the best-possible mechanisms are known \cite{procaccia2013approximate}, significant gaps remain in higher dimensions, even for $\R^2$.

For the \emph{utilitarian} objective of minimizing the social cost (sum of agents' costs), the \emph{coordinate-wise median (CWM)} mechanism is known to be the optimal deterministic strategyproof mechanism in~$\R^2$. Recently, \citet{barak2026facilitylocationmechanismdesign} showed how to break this deterministic barrier by designing a randomized mechanism (Randomly Rotated CWM or \emph{RR-CWM}) with a strictly better approximation ratio for $\R^2$, and showed that the same mechanism can be extended to $\R^d$.

However, for the \emph{egalitarian} objective of minimizing the maximum cost, the picture is quite different. In $\R$, \citet{procaccia2013approximate} give a strategyproof in expectation mechanism that achieves an approximation ratio of $1.5$. Beyond the one-dimension case,
no known randomized mechanism achieves a better approximation ratio (up to $o_n(1)$ terms) than the straightforward 
deterministic dictator mechanism, which achieves a ratio of $2$~\cite{gonzalez1985clustering,AlonFPT10}, even in $\R^2$ \cite{balkanski2024randomized}. The best known lower bound in $\R^2$ is only $1.118$ \cite{balkanski2024randomized}, leaving a significant gap between the known upper and lower bounds.%
\footnote{A concurrent work \cite{gomes2026improved} has improved this lower bound to $\approx 1.577$. See \cref{sec:related-work} for a comparison to our results.}
In this work, we focus on the power of randomized mechanisms for the egalitarian objective. 

\paragraph{Ex-ante versus ex-post guarantees.}

There is more than one natural way to measure solution quality for randomized mechanisms addressing the egalitarian objective. Since we consider mechanisms that are only strategyproof \emph{in expectation}, each agent already evaluates the outcome by her expected cost. Indeed, \citet{procaccia2013approximate} show that requiring universal strategyproofness (that is, having strategyproofness for every realization of the random coins used by the mechanism) implies an impossibility of achieving an approximation ratio of $2 - \Omega(1)$, while relaxing to strategyproofness in expectation allows for a $1.5$-approximation in $\R$. It is therefore natural to also evaluate the mechanism itself in terms of expected costs. Prior work has focused exclusively on the \emph{ex-post} guarantee of minimizing the expected maximum cost. Here, we advocate a complementary guarantee that may circumvent the impossibility results above: the \emph{ex-ante} guarantee of minimizing the maximum expected cost, i.e., ensuring that no single agent's expected cost is too high.

The egalitarian objective can also be interpreted through the lens of fairness: on one end, the (unfair) deterministic dictator mechanism achieves an approximation ratio of $2$; on the other end, an ideally fair mechanism would return the point in the middle of all locations (center of the minimum enclosing ball) and achieve a ratio of $1$. Under this interpretation, minimizing the expected maximum cost corresponds to \emph{expected ex-post fairness}, and minimizing the maximum expected cost corresponds to \emph{expected ex-ante fairness}; both well-studied notions in the fair division literature (see \cref{sec:related-work}).

Formally, let $\expostc$ denote the expected maximum cost (the ex-post egalitarian cost) and let $\exantec$ denote the maximum expected cost (the ex-ante egalitarian cost). By Jensen's inequality, for every instance $P$ and mechanism $M$ we have:
\[
    \expostc(P, M) \ge \exantec(P, M).
\]

Consider the following $2$-agent example:
\begin{example}
In $\R^2$, let one agent be located at $e_1 = (1,0)$ and the other at $-e_1 = (-1,0)$.
For the ex-post objective, the LRM mechanism (extended to $\R^2$) achieves an approximation ratio of $1.5$ \cite{procaccia2013approximate}, while the random dictator mechanism achieves an approximation ratio of~$2$. For the ex-ante objective, however, the random dictator mechanism is optimal on this instance, with an approximation ratio of $1$, as each agent's realized cost is $0$ w.p.\ $0.5$ and $2$ w.p.\ $0.5$, so each agent's expected cost is $1$.
\end{example}

This example illustrates why the ex-ante objective may allow us to gain better performance: randomizing between extreme locations may improve the expected cost of the agents. This motivates our first research question. Our second research question concerns the power of randomization more broadly.

\begin{question}\label{question:ex-ante-vs-ex-post}
    Can we get strictly better approximation ratios for the ex-ante objective than for the ex-post objective?
\end{question}

\begin{question}\label{question:randomization}
    May randomization help break the deterministic barrier of $2$ for the egalitarian objective in Euclidean space $\R^d$ for $d \ge 2$?
\end{question}

As we show in this work, the answers to these questions depend on the dimension of the space. In low dimensional space, the answer to both is yes: there is a gap between the ex-ante and ex-post objectives in both $\R$ and $\R^2$, and randomization helps break the deterministic barrier of $2$ in $\R^2$ for the ex-ante objective.

However, in high dimensional space ($\R^d$ with $d \gg 1$), we show that both the ex-ante and ex-post objectives have a lower bound of $2 - o_d(1)$, implying a strong impossibility result: \emph{even with randomization, no mechanism is more fair (has a significantly better egalitarian approximation) than any deterministic dictator mechanism}.

\subsection{Our results}\label{sec:contribution}

Our main results are summarized in \cref{tab:results}. For an overview of our techniques see~\Cref{sec:technical-overview}.
\begin{table}[h]
\centering
\small
\begin{tabular}{c c c c}
\hline
Objective & Dimension & Lower bound & Upper bound \\
\hline
\multirow{3}{*}{Ex-ante} & $1$ & $1$ & $1$ (\cref{thm:exante-LR-is-sp-and-optimal}) \\
 & $2$ & $1.09$ (\cref{thm:R2-lower-bounds}) & $1.598$ (\cref{thm:alg-rrc-sp-and-approx}) \\
 & $d$ & $2 - O\prn*{1/\sqrt{d}}$ (\cref{thm:Rd-lb-2}) & $2^*$ \\
\hline
\multirow{3}{*}{Ex-post} & $1$ & $3/2$~\cite{procaccia2013approximate} & $3/2$~\cite{procaccia2013approximate} \\
 & $2$ & $1.605$ (\cref{thm:R2-lower-bounds}) & $2^*$ \\
 & $d$ & $2 - O\prn*{1/\sqrt{d}}$ (\cref{thm:Rd-lb-2}) & $2^*$ \\
\hline
\end{tabular}
\caption{Approximation bounds for strategyproof in expectation mechanisms.\\$^*$Trivially achieved by any deterministic dictator mechanism~\cite{AlonFPT10,gonzalez1985clustering}.}
\label{tab:results}
\end{table}

Our first contribution is the introduction of the ex-ante objective, which is a natural alternative to the ex-post objective. We show that for the ex-ante objective, we can achieve better approximation guarantees than are possible for the ex-post objective, in both $\R$ and $\R^2$, thus answering \cref{question:ex-ante-vs-ex-post} affirmatively in low dimension.
In $\R$ we show that the simple LR mechanism (\cref{alg:lr}), which returns the leftmost or rightmost reported point each with probability $\frac12$, is strategyproof in expectation and \emph{optimal} for the ex-ante objective, achieving a ratio of exactly $1$ (\cref{thm:exante-LR-is-sp-and-optimal}). This is contrasted with the tight ratio of $1.5$ known for the ex-post objective~\cite{procaccia2013approximate}: a strict separation between the two objectives already in one dimension.
In $\R^2$ we design two mechanisms, $\rc$ and $\rrc$, which only ever return one of the (at most four) extreme corners of the reported points' minimum bounding box, possibly after applying a uniformly random rotation to the instance. Both mechanisms are strategyproof in expectation and achieve expected approximation ratios of $\phi \approx 1.618$ (\cref{thm:alg-rc-sp-and-phi-approx}) and $\approx 1.598$ (\cref{thm:alg-rrc-sp-and-approx}), respectively, for the ex-ante objective, beating the deterministic barrier of $2$, and thus also answering \cref{question:randomization} affirmatively in $\R^2$.

We complement these upper bounds with lower bounds in $\R^2$: we show that any strategyproof in expectation mechanism has an expected approximation ratio of at least $\approx 1.605$ for the ex-post objective, and at least $\approx 1.09$ for the ex-ante objective (\cref{thm:R2-lower-bounds}). In particular, this confirms that the ex-post objective is provably harder to approximate than the ex-ante objective in $\R^2$, and it also improves, for the ex-post objective, on the previously best known lower bound of $1.118$~\cite{balkanski2024randomized}.

Our high-dimensional lower bound is a strong impossibility result: In high-dimensional space ($\R^d$ with $d \gg 1$), we show that both the ex-ante and ex-post objectives have a lower bound of $2 - o_d(1)$ (\cref{thm:Rd-lb-2}), implying a strong impossibility result, showing that even with randomization one cannot do better than any deterministic dictator mechanism.
This shows that the answers to \cref{question:ex-ante-vs-ex-post,question:randomization} are dimension-dependent: both randomization and the ex-ante relaxation are powerful tools in low dimension, yet they provably lose all their power as the dimension grows, revealing an interesting separation between low- and high-dimensional Euclidean space for the egalitarian objective.

Finally, we show that the currently best known mechanism for the utilitarian objective of minimizing the expected sum of agent costs, $\rrcwm$, is also a good mechanism for the egalitarian objective, achieving an approximation ratio of $2$ for both objectives in $\R^d$ for any $d \ge 1$ (\cref{thm:rrcwm-max-cost}). This shows that one can achieve a good approximation for both objectives simultaneously, getting the best of both worlds. We note that, perhaps surprisingly, the simpler deterministic variant of $\rrcwm$, the $\cwmed$ mechanism, is not a good mechanism for the egalitarian objective, as it has an approximation ratio of $1 + \sqrt{2} \approx 2.41$.

\subsection{Further related work}\label{sec:related-work}

\paragraph{Egalitarian facility location mechanism design.}
On the real line, there's a well-known deterministic lower bound of 2 \cite[Theorem 3.2]{procaccia2013approximate}, which is extended to multidimensional $L_p$ spaces \cite[Lemma 7]{lin2020nearly}. The matching upper bound is attained by any \emph{dictator mechanism}, which fixes an agent $i$ in advance and always locates the facility at her reported location $x_i$; it is group-strategyproof and has approximation ratio $2$ for the maximum cost in every metric space \citep{AlonFPT10}. The \emph{random dictator mechanism} chooses $x_i$ with probability $1/n$ for every agent $i$. It is a report-independent lottery over dictator mechanisms and is therefore universally strategyproof; every realized outcome retains the factor-$2$ maximum-cost guarantee \citep{AlonFPT10}. On the line, a $p$-\emph{percentile mechanism} returns the $\prn*{\lfloor (n-1)p \rfloor+1}$-th ordered report. This group-strategyproof family includes the leftmost, median, and rightmost mechanisms for $p=0$, $p=1/2$, and $p=1$, respectively; its multidimensional extension selects a fixed percentile independently in each coordinate and remains strategyproof \citep{sui2013analysis}. Finally, the \emph{centroid mechanism} returns the centroid $\bar{x}=\frac{1}{n}\sum_{i=1}^n x_i$ with probability $1/2$ and each reported location $x_i$ with probability $1/(2n)$ \citep{FeldmanW13}. In Euclidean space, it is strategyproof in expectation and achieves a tight $2-1/n$ approximation for the egalitarian ex-post objective \citep{TangYZ20}.

The best known lower bound is $1.118$ \cite{balkanski2024randomized} in $\R^2$, while for trees there exists a lower bound of $2 - o(1)$ \cite{AlonFPT10}. 

Concurrent independent work by \citet{gomes2026improved} studies the ex-post objective and improves the lower bound in $\R^2$ to $\approx 1.577$. Their construction is similar to the one we give for $\R^2$, yet it is different: we both use a transition from simplex to circle, but we start from two points, which changes the approximation ratio lower bound. Indeed, our construction leads to a slightly better bound of $\approx 1.605$ for the ex-post objective (and $\approx 1.09$ for the ex-ante objective which they do not study). Their construction also generalizes to $\R^d$ for $d>2$ and leads to a lower bound of $1 + \sqrt{\frac{d}{2(d+1)}}$ which goes to $1 + \frac{1}{\sqrt{2}} \approx 1.707$ as $d \to \infty$. In comparison, our recursive $\R^d$ construction leads to a lower bound of $2 - O(\frac{1}{\sqrt{d}})$ which goes to $2$ as $d \to \infty$ for both the ex-post and ex-ante objectives.
The second similarity concerns the $\sqrt{2}$ upper bound for $2$ agents in $\R^2$ which we provide in Appendix~\ref{sec:2-agents-upper-bound}, as intuition for why the two-agent one-dimensional lower bound does not transfer to $\R^2$. They give a similar $\sqrt{2}$-approximate mechanism, generalized to $d$ dimensions. They also study the case of a small number of agents, and output augmentation (where the agent locations are restricted to a subset of the space but the returned facility location is not).

Another concurrent independent work of \cite{hastings2026maximumcoststrategicfacilitylocation} studies the ex-post egalitarian objective and showed a different $2 - o_1(d)$ lower bound for $\R^d$.

% The facility location mechanism design problem was also studied in the context of learning augmented algorithms, where learning-augmented \emph{mechanism design} was initiated by \citet{agrawal2022learning,ijcai2022p81} and studied further for facility location \citep{agrawal2022learning,ijcai2022p81,istrate2022mechanism,chen2024strategic,balkanski2024randomized,barak2024mac,Shi2025PredictionAugmented,walsh2025mechanism,gravin2025approximation, barak2026facilitylocationmechanismdesign}.

\paragraph{Other metric spaces.}
\citet{AlonFPT10} give a $2 - o(1)$ lower bound for trees. While the construction of \cite{AlonFPT10} also follows a similar structure of recursively building the ``bad'' instance and amplifying the expected distance at each step, the same construction does not work in Euclidean space. The reason is that in trees distances are additive, and in $\R^d$ they're not: in trees one must use tree edges to get from point $A$ to $B$ (potentially going through other tree vertices), whereas in $(\R^d,\ell_2)$ one could go directly from $A$ to $B$. Other structures were also considered. \citet{FeldmanW13} also study trees for a different objective of minimizing the sum of squared distances.
\citet{AlonFPT10} also study rings and more generally network (general graph) structures (which were also studied by \cite{schummer2002strategy}).

For the \emph{utilitarian} objective the median is both strategyproof and optimal in $\R$ \citep{moulin1980strategy,ProccacciaTennenholtz2009,procaccia2013approximate}. In $\R^2$, $\cwm$ achieves an approximation ratio of $\sqrt{2}$, which is optimal among deterministic strategyproof mechanisms \citep{durocher2009projection,goel2023optimality}.

\paragraph{More than one facility.}
Several works study the opening of $k\ge 2$ facilities, including \cite{procaccia2013approximate,lu2009,LuSWZ10,FotakisT10,EscoffierGTPS11,FotakisT16,walsh2020strategy,barak2024mac,ma2026breaking4approximationbarrierstrategyproof}.

\paragraph{Improvements via random rotations.}
Random rotations were studied in the past for the different \emph{utilitarian} objective.
\citet{meir2019strategyproof} and \citet{goel2023optimality} conjectured that random rotations of the axes may help decrease the expected approximation ratio, and 
\citet{barak2026facilitylocationmechanismdesign} showed that for the utilitarian objectives rotations may help break the deterministic barrier in the classic setting as well as the learning-augmented setting. Their mechanism randomly rotates the points, computes the coordinate-wise median and rotates the result back. \citet{chan2026strategyproofmechanismseuclideanfacility} also give similar results for other $L_p$ objectives.
\citet{Gershkov2019} studied deterministic rotations and showed that given certain distributional assumptions (e.g., identical marginals satisfying additional regularity) a fixed $\pi/4$ rotation followed by the coordinate-wise median can outperform no rotation.

\paragraph{Ex-ante versus ex-post in voting and fairness.}
For lotteries over discrete outcomes, ex-ante fairness is evaluated on expected utilities or the induced fractional outcome, whereas ex-post fairness must hold in every realization. In random assignment, this separates probabilistic serial, which is envy-free and ordinally efficient, from random priority, which is strategyproof and ex-post efficient but may fail envy-freeness and ordinal efficiency \citep{BogomolnaiaMoulin2001}. Best-of-both-worlds fair division combines exact ex-ante guarantees with relaxed ex-post ones, such as EF with EF1 or proportionality with MMS- and EFX-type guarantees; stronger combinations are known for few agents \citep{AzizFreemanShahVaish2024,BabaioffEzraFeige2022,BabaioffFrosh2026}. Analogously, work on probabilistic and committee voting combines ex-ante fair-share axioms such as IFS, UFS, GFS, and GRP with ex-post representation axioms such as EJR and FJR \citep{BogomolnaiaMoulinStong2005,AzizLuSuzukiVollenWalsh2023,SuzukiVollen2024}. Our objectives instead apply the two orders to one egalitarian cost, $\max_i \E\brk*{c_i}$ versus $\E\brk*{\max_i c_i}$; thus our ex-post criterion is an expected realized maximum, not a support-wise fairness requirement.

\subsection{Paper organization}
In \cref{sec:preliminaries} we give preliminaries. In \cref{sec:technical-overview} we provide a technical overview. In \cref{sec:lb-2-Rd} we give our $2 - o_d(1)$ lower bound for $\R^d$, which applies to both objectives. In \cref{sec:exantec-ub} we give our upper bounds for the ex-ante objective. In \cref{sec:lb-R2} we give our lower bounds in $\R^2$ for both the ex-ante and the ex-post objectives. In \cref{sec:egalitarian-approx-of-rrcwm} we analyze the egalitarian approximation of the $\rrcwm$ mechanism. We conclude with a discussion in \cref{sec:discussion}.

\section{Preliminaries}
\label{sec:preliminaries}

\paragraph{The egalitarian facility location mechanism design problem.}

A mechanism $M$ receives as input $P$, a set of $n$ locations of strategic agents in Euclidean space $\R^d$, for $n,d \in \N$. The output of the mechanism is a (random) location $Y_P := M(P) \in \R^d$ for the facility. For any profile $P$, let $ALG(P)$ denote the cost of the mechanism and $OPT(P)$ the cost of an optimal solution\footnote{Sometimes we drop $P$ from the notation where the profile is clear from context and simply write $ALG$ and $OPT$.}. 
The \emph{cost}
of each agent is the $\ell_2$ distance from the agent location to the returned facility.
The \emph{approximation ratio} $\alpha(M)$ of a random mechanism $M$ is defined as the worst-case ratio over all instances $P$ between the cost $ALG(P)$ of the mechanism and that of an optimal solution, denoted by $OPT(P)$.
Formally:
\[
    \alpha(M) := \sup_{P \subset \R^d} \frac{ALG(P)}{\OPT(P)},\footnote{If $\OPT(P) = 0$, we define the approximation ratio to be $1$ if $ALG(P) = 0$ and $\infty$ otherwise.}
\]

The goal is to design a strategyproof mechanism $M$ that returns a (random) location that minimizes the approximation ratio, where strategyproofness is defined as follows.
Let $P = (p_1, \ldots ,p_n) \subset \R^d$ be the locations of $n$ agents in space, and for any agent $i \in [n]$ let $(P_{-i},y)$ be the vector resulting from replacing $p_i$ in $P$ %where $p_i$ is replaced 
by $y \in \R^d$.

\begin{definition}[Strategyproofness]
\label{def:strategyproofness}
    A mechanism $M$ is \emph{strategyproof in expectation} if for any agent $i \in [n]$ located at $p_i$, the agent's expected cost cannot decrease by reporting a location $p'_i \neq p_i$, that is, $\E\brk*{\norm{p_i - M(P)}_2} \le \E\brk*{\norm{p_i - M\prn*{P_{-i},p'_i}}_2}$. A mechanism $M$ is \emph{universally strategyproof} if this holds for every realization of $M$'s coin tosses, i.e., $\norm{p_i - M(P)}_2 \le \norm{p_i - M\prn*{P_{-i},p'_i}}_2$; universal strategyproofness implies strategyproofness in expectation.
\end{definition}

Throughout the paper, ``strategyproof'' refers to strategyproofness in expectation unless stated otherwise.
Next, we give the formal definition of both the ex-ante and ex-post objectives for random mechanisms.

\begin{definition}[Ex-Ante Egalitarian Cost]
\label{def:ex-ante-cost}
    The ex-ante egalitarian cost of a dataset $P \subset \R^d$ and a random mechanism $M$ is the maximum expected distance between an agent and the output of $M$. That is: $\exantec(P,M) = \max_{p_i \in P} \E\brk*{\norm{p_i - M(P)}_2}$.
\end{definition}

\begin{definition}[Ex-Post Egalitarian Cost]\label{def:ex-post-cost}
    The ex-post egalitarian cost of a dataset $P \subset \R^d$ and a random mechanism $M$ is the expected maximum distance between an agent and the output of $M$. That is: $\expostc(P,M) = \E\brk*{\max_{p_i \in P} \norm{p_i - M(P)}_2}$.
\end{definition}

\paragraph{Other notation.} 
proofs shorter and more clear.
Let $M$ be a random mechanism. For any facility location instance $P \subset \R^d$, let $Y_P := M(P)$ denote the (random) point returned by $M$ on input $P$. For any point $u \in \R^d$, let $D_{P}(u) := \E\brk*{\norm{Y_P - u}_2}$ denote the expected distance of the output of $M$ from $u$. For a finite set $S$, let $\mathrm{Unif}(S)$ denote the uniform distribution over $S$.

\section{Technical overview}\label{sec:technical-overview}

\paragraph{Ex-ante upper bounds via extreme points.}
For our ex-ante upper bounds we use the four indexed corners of the reports' minimum bounding box. Fixing the other agents' reports, truth-telling makes each lower and upper coordinate endpoint individually as close as possible to the agent's true coordinate. We couple every truthful corner with the deviating corner that uses the same lower/upper choice in each coordinate; both coordinate differences, and hence the Euclidean distance, can only increase under the deviation. Averaging over the four matched corners proves strategyproofness in expectation, and the same argument applies after an independently sampled rotation. Bounding the approximation ratio then reduces to a geometric optimization problem: we bound the worst-case expected distance from a point in the minimum enclosing ball to a uniformly random corner (or, for $\rrc$, to a randomly rotated corner) of the bounding box. Convexity of the resulting expected-distance function, together with the Bauer maximum principle, lets us restrict attention to the boundary of the enclosing ball, where the bound follows from an explicit (and, for $\rrc$, trigonometric) calculation.

\paragraph{Lower bounds via distance amplification.}
Both our $\R^d$ and $\R^2$ lower bounds are built from a single common technique of \emph{distance amplification}. We build a hard instance via a sequence of steps, where each step moves already-placed agents, one at a time, from a heavily populated point (where agents are colocated at the same place) to a small set of new locations arranged in a symmetric structure around it. The exact way to use distance amplification differs as a function of the studied space. Strategyproofness in expectation guarantees that moving agents this way cannot decrease the expected distance of the mechanism's output from the point being vacated, since otherwise an agent there could profitably misreport to one of the new locations. Each such structure comes with a pointwise geometric inequality lower-bounding the average (or maximum) distance from its points to an arbitrary point, in terms of the distance from that point to the structure's center; convexity of this bound together with Jensen's inequality then lets us propagate and amplify a growing expected-distance guarantee across steps. The $\R^d$ and $\R^2$ constructions differ in exactly which symmetric structures they chain together, and why.

\paragraph{High-dimensional lower bound via recursive orthogonal subspace simplex amplification.}
\begin{figure}[h]
\centering
\begin{tikzpicture}[
    every node/.style={font=\small},
    pt/.style={circle, fill=black, inner sep=0pt, minimum size=4pt},
    altpt/.style={circle, fill=gray!45, inner sep=0pt, minimum size=3pt},
]
    \coordinate (x0) at (0,0);
    \coordinate (x1) at (1.3,0.35);
    \coordinate (x2) at (1.65,1.65);
    \coordinate (x3) at (2.95,2.0);
    \coordinate (xk) at (4.35,3.7);
    \coordinate (mid) at ($(x0)!0.5!(xk)$);

    % enclosing ball
    \draw[dashed, gray] (mid) circle (3.15);

    % alternative (unchosen) simplex vertices at each step - faint fan
    \node[altpt] at ($(x0)+(0.95,-0.55)$) {};
    \node[altpt] at ($(x0)+(0.35,0.95)$) {};
    \node[altpt] at ($(x1)+(-0.25,1.15)$) {};
    \node[altpt] at ($(x1)+(0.95,0.85)$) {};
    \node[altpt] at ($(x2)+(0.75,-0.55)$) {};
    \node[altpt] at ($(x2)+(1.55,0.9)$) {};

    % path
    \draw[-{Stealth[length=2mm]}, thick] (x0) -- (x1);
    \draw[-{Stealth[length=2mm]}, thick] (x1) -- (x2);
    \draw[-{Stealth[length=2mm]}, thick] (x2) -- (x3);
    \draw[-{Stealth[length=2mm]}, thick, dashed] (x3) -- (xk);
    \node[fill=white, inner sep=1pt, rotate=50.6] at ($(x3)!0.5!(xk)$) {\footnotesize $\cdots$};

    % points
    \node[pt] at (x0) {};
    \node[pt] at (x1) {};
    \node[pt] at (x2) {};
    \node[pt] at (x3) {};
    \node[pt] at (xk) {};

    % labels: point name + amplified distance guarantee
    \node[below left=1pt of x0] {\shortstack{$x_0$\\ {\color{purple}\scriptsize $\alpha_0{=}0$}}};
    \node[below right=1pt of x1] {\shortstack{$x_1$\\ {\color{purple}\scriptsize $\alpha_1$}}};
    \node[above left=1pt of x2] {\shortstack{$x_2$\\ {\color{purple}\scriptsize $\alpha_2$}}};
    \node[below right=1pt of x3] {\shortstack{$x_3$\\ {\color{purple}\scriptsize $\alpha_3$}}};
    \node[above right=1pt of xk] {\shortstack{$x_k$\\ {\color{purple}\scriptsize $\alpha_k\gtrsim\sqrt k$}}};

    % coordinate-block labels
    \node[gray, font=\scriptsize] at ($(x0)!0.5!(x1)+(0.15,-0.3)$) {step $1$};
    \node[gray, font=\scriptsize, anchor=west] at ($(x1)!0.5!(x2)+(0.2,0)$) {step $2$};
    \node[gray, font=\scriptsize] at ($(x2)!0.5!(x3)+(0.15,-0.3)$) {step $3$};

    % radius label
    \node[gray] at ($(mid)+(-2.05,2.35)$) {\scriptsize radius $\approx \sqrt{k}/2$};

\end{tikzpicture}
\caption{The recursive orthogonal-block simplex construction underlying \cref{thm:Rd-lb-2}: at step $t$, the agent cluster at $x_t$ is split among the vertices of a unit simplex spanning a new, previously unused $d'$-dimensional coordinate block (at step $t$; some unselected simplex vertices are shown in gray), and the vertex $x_{t+1}$ with amplified guaranteed expected distance $\alpha_{t+1}$ is selected as the next point. After $k \approx \sqrt d$ such mutually orthogonal steps, $\alpha_k \gtrsim \sqrt k$, while the entire path remains inside a ball of radius $\approx \sqrt k/2$.}
\label{fig:rd-construction}
\end{figure}
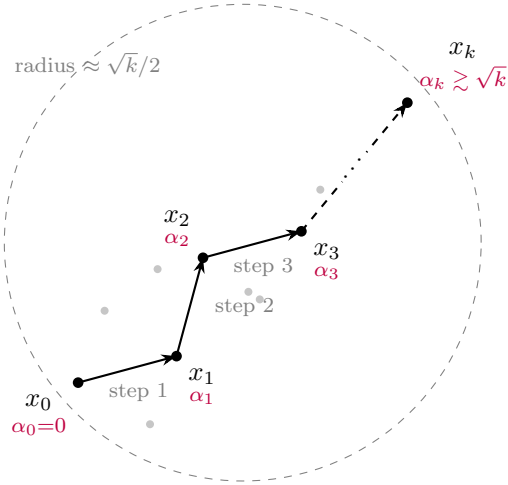
For our high dimensional lower bound of $2 - o_d(1)$, we partition the $d$ coordinates into $k \approx \sqrt d$ orthogonal blocks of $d' \approx \sqrt d$ coordinates each, and construct the instance in $k$ amplification steps (\cref{fig:rd-construction}): at step $t$ we split the agents currently located at a point $x_t$ equally among the vertices of a unit regular simplex spanning the $t$-th coordinate block, moving them there one at a time. A simplex-averaging inequality shows that at least one simplex vertex $x_{t+1} = x_t + \hat s^t$ has expected distance from the mechanism's output that is amplified via $\alpha_{t+1} \ge \sqrt{\alpha_t^2 + 1 - \frac{4}{d'}}$, where $\alpha_t$ denotes the guaranteed expected distance from $x_t$ after step $t$. Since the $k$ simplex directions used across steps are mutually orthogonal, after $k$ steps we obtain $\alpha_k \gtrsim \sqrt{k}$, while every reported point remains within a ball of radius $\approx \sqrt k / 2$ around the point $\frac{x_{k-1}}{2}$. Taking $k \approx \sqrt d$ then yields the claimed ratio of $2 - O(1/\sqrt d)$.

\paragraph{$\R^2$ lower bounds.}
The high-dimensional construction cannot be used for $\R^2$. It relies on chaining $k \approx \sqrt d$ mutually \emph{orthogonal} amplification steps, so that the guarantee $\alpha_k$ grows like $\sqrt k$ while the instance stays confined to a ball whose radius also grows only like $\frac{1}{2}\sqrt k$; this requires the ambient dimension to grow with $k$, and degenerates to a trivial bound once the dimension is fixed and small, as in $\R^2$. A more naive fix, simply reusing the one-dimensional two-point construction underlying the tight $3/2$ ex-post lower bound of \citet{procaccia2013approximate} in $\R$, also fails: as \citet{balkanski2024randomized} observe, it only yields a lower bound of $\approx 1.118$ in $\R^2$. Indeed, we show in \cref{sec:2-agents-upper-bound} that for two agents there is a strategyproof in expectation mechanism that is $\sqrt{2}$-approximate for both objectives, since in $\R^2$ the mechanism may not only return a point on the segment connecting the two agents' locations, but also ``think outside the box''. We therefore turn to a different, two-dimensional structure.

\begin{figure}[h]
\centering
\begin{tikzpicture}[
    every node/.style={font=\small},
    pt/.style={circle, fill=black, inner sep=0pt, minimum size=4pt},
]

% Panel A: two colocated groups of agents
\begin{scope}
    \coordinate (u) at (-0.8,0);
    \coordinate (v) at (0.8,0);
    \draw[dashed, gray] (u) -- (v);
    \node[pt] at (u) {};
    \node[pt] at (v) {};
    \node[below=2pt of u] {\shortstack{$u$\\ {\color{purple}\scriptsize $2k$ agents}}};
    \node[below=2pt of v] {\shortstack{$v$\\ {\color{purple}\scriptsize $2k$ agents}}};
    \node[below=32pt] at (0,-0.9) {\shortstack{(a) Two colocated\\ groups of agents}};
\end{scope}

% Arrow 1
\draw[-{Stealth[length=2.5mm]}, thick] (1.3,0) -- (2.7,0)
    node[midway, above] {\scriptsize triangle step};

% Panel B: equilateral triangle
\begin{scope}[xshift=4cm]
    \coordinate (c) at (0,0);
    \coordinate (w0) at (0,0.95);
    \coordinate (w1) at (-0.82,-0.48);
    \coordinate (w2) at (0.82,-0.48);
    \draw[dashed, gray] (c) -- (w0);
    \draw[dashed, gray] (c) -- (w1);
    \draw[dashed, gray] (c) -- (w2);
    \draw (w0) -- (w1) -- (w2) -- cycle;
    \node[pt] at (c) {};
    \node[pt] at (w0) {};
    \node[pt] at (w1) {};
    \node[pt] at (w2) {};
    \node[above=2pt of w0] {\shortstack{{\color{purple}\scriptsize $2k$ agents}\\$w_0$}};
    \node[below left=1pt of w1] {\shortstack{$w_1$\\ {\color{purple}\scriptsize $k$ agents}}};
    \node[below right=1pt of w2] {\shortstack{$w_2$\\ {\color{purple}\scriptsize $k$ agents}}};
    \node[right=3pt of c] {$c$};
    \node[below=32pt] at (0,-0.9) {\shortstack{(b) Equilateral\\ triangle}};
\end{scope}

% Arrow 2
\draw[-{Stealth[length=2.5mm]}, thick] (5.3,0) -- (6.7,0)
    node[midway, above] {\scriptsize circle step};

% Panel C: discretized circle around a triangle vertex
\begin{scope}[xshift=8cm]
    \coordinate (b) at (0,0);
    \draw (b) circle (0.9);
    \foreach \i in {0,...,9} {
        \coordinate (z\i) at ({0.9*cos(\i*36)},{0.9*sin(\i*36)});
        \node[pt, minimum size=3pt] at (z\i) {};
    }
    \draw[dashed, gray] (b) -- (z0) node[midway, above] {\scriptsize $D$};
    \node[pt] at (b) {};
    \node[below right=0pt of b] {$b$};
    \node[right=1pt of z0] {\shortstack{$z_0$\\ {\color{purple}\scriptsize $1$ agent}}};
    \node[below right=-1pt of z9] {\shortstack{$z_{k-1}$\\ {\color{purple}\scriptsize $1$ agent}}};
    \node[below=32pt] at (0,-0.9) {\shortstack{(c) Discretized\\ circle around $b$}};
\end{scope}

\end{tikzpicture}
\caption{The three-step distance-amplification construction underlying \cref{thm:R2-lower-bounds} (nonzero agent counts shown in purple): two colocated groups of $2k$ agents are split into an equilateral triangle $w_0,w_1,w_2$ centered at $c$ (now empty), and then $k$ agents at one triangle vertex $b\in\crl*{w_0,w_1,w_2}$ are selected and split among $k$ equally-spaced points $z_0,\ldots,z_{k-1}$ on a circle of radius $D$ around $b$.}
\label{fig:r2-construction}
\end{figure}
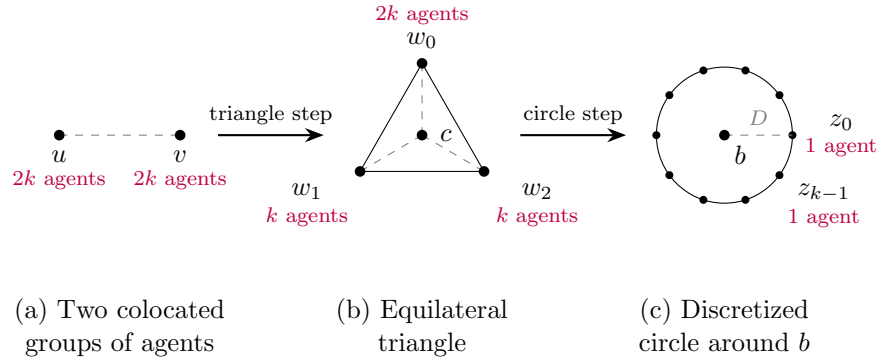
For our $\R^2$ lower bounds, both for the ex-ante and the ex-post objectives, we build a single hard instance via two amplification steps, using a different symmetric structure at each (\cref{fig:r2-construction}): we start with two points on the line just like \cite{procaccia2013approximate}, then move the agents at one of the two points to the vertices of an equilateral triangle around the point, and finally transition to points evenly spaced on a circle centered at one of the triangle's vertices. The triangle step amplifies a distance guarantee once, and the circle step further amplifies this into our bounds: an ex-ante bound via the average distance to a random point on the circle, and a stronger ex-post bound via the distance to the point on the circle antipodal to the mechanism's output.

\paragraph{Best of both worlds via the random rotation coordinate-wise median.}
To analyze $\rrcwm$ for the egalitarian objective (\cref{sec:egalitarian-approx-of-rrcwm}), we center the reports at their minimum enclosing ball and view each coordinate median after a uniformly random rotation of the points as the median $F(U)$ of the projections of the agents' reported locations onto a uniformly random direction $U$. The scalar one-dimensional median is $1$-Lipschitz in the $\ell_\infty$ norm, so $F$ is Lipschitz with constant equal to the enclosing-ball radius; the Gaussian Poincar\'e inequality is then used to give a bound of $(r^*)^2/d$ in each direction. Summing this bound over the marginally uniform rows of the rotation shows that the squared distance between the output and the optimal center is at most $(r^*)^2$ in expectation, and Jensen's and the triangle inequalities yield the factor of $2$ ex-post guarantee. We contrast this with the deterministic $\cwmed$: a coordinate-counting argument bounds its distance from the optimal center by $\sqrt{2}r^*$, giving the tight dimension-independent factor $1+\sqrt{2}$.

\section{Lower bound of $2 - o_d(1)$ for $\R^d$}\label{sec:lb-2-Rd}

In this section we show a lower bound of $2 - o_d(1)$ for both objectives.
This implies that no strategyproof in expectation mechanism is better than any deterministic dictator mechanism (up to $o_d(1)$) in high dimensional space.

\begin{thm}\label{thm:Rd-lb-2}
    For every $d\ge 25$, any strategyproof in expectation mechanism for $\R^d$ has an expected approximation ratio of at least $2 - O\prn*{\frac{1}{\sqrt{d}}}$ for both the ex-ante and ex-post objectives.
\end{thm}
It suffices to establish the lower bound for the maximum-expected-cost (ex-ante)
objective. Indeed, by Jensen's inequality, $\E\brk*{
\max_{i\in[n]}\lVert Y_P-p_i\rVert_2
}
\ge
\max_{i\in[n]}
\E\brk*{
\lVert Y_P-p_i\rVert_2
}$.
Moreover, the optimal value for both objectives is the radius of the
minimum enclosing ball of $P$. Hence, any approximation lower bound
for the maximum-expected-cost (ex-ante) objective also applies to the
expected-maximum-cost (ex-post) objective.

\paragraph{Proof idea.}
Our construction starts with a large cluster of $n = (k+1)^k$ agents, all co-located at the origin.
The construction proceeds for $k$ steps. At step $t$, a distinguished
cluster $C_t$ is located at a point $x_t$, and we have already guaranteed $\E\brk*{d(x_t,Y_{P_t})}\ge\alpha_t$ for $\alpha_t$.
We split $C_t$ equally among the $d'+1$ vertices of a unit regular
simplex lying in a new $d'$-dimensional coordinate block around $x_t$.\footnote{We could similarly use a sphere or axis aligned rectangle instead, but $d'+1$ simplex points are enough.}

Strategyproofness implies that moving these agents away from $x_t$,
one at a time, cannot decrease the expected distance of the outcome
from $x_t$. A simplex averaging inequality then guarantees that at least one child $x_t+s_t$ satisfies
\[
\alpha_{t+1}
:=
\E\brk*{d(x_t+s_t,Y_{P_{t+1}})}
\ge
\sqrt{\alpha_t^2+1-\frac{4}{d'}},
\]
and we select this child as $x_{t+1}$.

We show that this expected distance between the randomly returned point and $x_t$ grows to approximately $\sqrt{k}$, and at the same time, every
reported point lies in a ball of radius at most $\approx \frac{\sqrt{k}}{2}$. The resulting ratio therefore approaches $2$.

\paragraph{Expected distance amplification via the subspace simplex.}
Before showing the proof of \cref{thm:Rd-lb-2}, we first show a useful lemma. The lemma roughly shows that given a subspace $V$ of $\R^d$ of dimension $4<d'< d$ and a point $x \in \R^d$, for every random point $Y$, there's a point $x'$ around $x$ such that the expected distance between $Y$ and $x'$ is larger than the one between $Y$ and $x$. The point $x'$ is obtained by moving from $x$ to one of the vertices of a unit simplex in $V$ around $x$.

For $w\in\mathbb R^d$, let $w_V$ denote the component of $w$ lying in
$V$. Thus $w=w_V+w_\perp$,
where $w_\perp$ is perpendicular to every vector in $V$. In the
coordinate-block subspaces used later, $w_V$ is obtained simply by
keeping the coordinates belonging to that block and setting all other
coordinates equal to zero.

A \emph{unit regular simplex} in a $d'$-dimensional subspace $V$ is a set $S=\{s_0,\ldots,s_{d'}\}\subseteq V$ of $d'+1$ unit vectors spanning $V$ and satisfying $\langle s_i,s_j\rangle=-\frac{1}{d'}$ for $i \neq j$.
Such a simplex exists in every $d'$-dimensional Euclidean space. First, we state some known identities of the simplex.
\begin{lem}\label{lem:simplex-identities}
    Let $S=\{s_0,\ldots,s_{d'}\}$ be a unit regular simplex in $V$. Then:
    (i) $\sum_{i=0}^{d'}s_i=0$, and
    (ii) for every $w\in\mathbb R^d$,
\[
\frac{1}{d'+1}
\sum_{i=0}^{d'}
\langle w,s_i\rangle^2
=
\frac{\lVert w_V\rVert_2^2}{d'}.
\]

\end{lem}
The proof of \cref{lem:simplex-identities} is rather standard and straightforward. We provide it in \cref{sec:missing-analysis-lb-2-Rd} for the sake of completeness.

\begin{lem}[Simplex expected distance amplification]
\label{lem:simplex-amplification}
Let $d'>4$, let $V\subseteq\mathbb R^d$ be a $d'$-dimensional subspace, and let $S=\{s_0,\ldots,s_{d'}\}\subseteq V$
be a unit regular simplex. Then, for every fixed $x\in\mathbb R^d$ and every
random point $Y$, there exists $s\in S$ such that
\[
\E
\brk*{
\left\lVert
Y-(x+s)
\right\rVert_2
}
\ge
\sqrt{
\left(
\E[\lVert Y-x\rVert_2]
\right)^2
+
1-\frac{4}{d'}
}.
\]
\end{lem}
Fix $d'>4$, a $d'$-dimensional subspace $V\subseteq\mathbb R^d$, and a unit regular simplex
$S=\{s_0,\ldots,s_{d'}\}\subseteq V$. For $x,y\in\mathbb{R}^d$, define the average distance from $y$ to the
simplex children of $x$ by
\[
\gamma(y)
:=
\frac{1}{d'+1}
\sum_{i=0}^{d'}
\left\lVert y-(x+s_i)\right\rVert_2.
\]

\begin{lem}[Simplex Average Distance Amplification]
\label{lem:simplex-avg-dist-amplification}
For every $x,y\in\mathbb{R}^d$,
\[
\gamma(y)
\ge
\sqrt{
\lVert y-x\rVert_2^2
+
1-\frac{4}{d'}
}.
\]
\end{lem}

\begin{proof}
Let $u=y-x$ and $\rho^2=\lVert u\rVert_2^2+1$ (so $\rho := \sqrt{\lVert u\rVert_2^2+1}$).
For every $i\in\{0,\ldots,d'\}$, we define $\delta_i := -\frac{2\langle u,s_i\rangle}{\rho^2}$.
By the fact that $\norm{s_i}_2 = 1$ we get: $2\lvert\langle u,s_i\rangle\rvert
\le
2\lVert u\rVert_2
\le
\lVert u\rVert_2^2+1
=
\rho^2$, where the first inequality follows from the Cauchy--Schwarz inequality and the second one since $(\norm{u} - 1)^2 \ge 0$. Thus, we have $\delta_i\in[-1,1]$.
Moreover,
\begin{align*}
\lVert u-s_i\rVert_2
&=
\sqrt{
\lVert u\rVert_2^2
+
\lVert s_i\rVert_2^2
-
2\langle u,s_i\rangle
}\\
&=
\sqrt{
\lVert u\rVert_2^2+1-2\langle u,s_i\rangle
}\\
&= \rho\sqrt{1+\delta_i}. \numberthis \label{eq:u-minus-s_i}
\end{align*}

A standard simple analysis shows that for every $z\in[-1,1]$: $\sqrt{1+z}
\ge
1+\frac{z}{2}-\frac{z^2}{2}$. Together with \cref{eq:u-minus-s_i} it implies that
\begin{align*}
\gamma(y)
&= \frac{1}{d'+1}
\sum_{i=0}^{d'}
\left\lVert u-s_i\right\rVert_2 =
\frac{\rho}{d'+1}
\sum_{i=0}^{d'}\sqrt{1+\delta_i}\\
&\ge
\rho
\left(
1
+
\frac{1}{2(d'+1)}
\sum_{i=0}^{d'}\delta_i
-
\frac{1}{2(d'+1)}
\sum_{i=0}^{d'}\delta_i^2
\right). \numberthis \label{eq:gamma-y-lb-1}
\end{align*}

Because the simplex is zero centered ($\sum_{i=0}^{d'}s_i = 0$ by \cref{lem:simplex-identities}),
\[
\sum_{i=0}^{d'}\delta_i
=
-\frac{2}{\rho^2}
\left\langle
u,\sum_{i=0}^{d'}s_i
\right\rangle
=
0.
\]
Furthermore, by the definition of $\delta_i$ and \cref{lem:simplex-identities}:
\[
\begin{aligned}
\frac{1}{d'+1}
\sum_{i=0}^{d'}\delta_i^2
&=
\frac{4}{\rho^4}
\cdot
\frac{1}{d'+1}
\sum_{i=0}^{d'}
\langle u,s_i\rangle^2 =
\frac{4\lVert u_V\rVert_2^2}{d'\rho^4}.
\end{aligned}
\]
Therefore, by plugging the last two equalities in \cref{eq:gamma-y-lb-1} we get:
\[
\gamma(y)
\ge \rho\prn*{1 - \frac{2 \norm{u_V}_2^2}{d' \rho^4}} =
\rho-
\frac{2\lVert u_V\rVert_2^2}{d'\rho^3}. \numberthis \label{eq:gamma-tmp-2}
\]

The right-hand side is nonnegative. Indeed, $\frac{2\lVert u_V\rVert_2^2}{d'\rho^4}
\le
\frac{2\lVert u\rVert_2^2}
{d'(\lVert u\rVert_2^2+1)^2}
\le
\frac{2}{d'}
<1$.
We may therefore square the preceding inequality~\eqref{eq:gamma-tmp-2} and get:
\begin{align*}
\gamma(y)^2
&\ge
\left(
\rho-
\frac{2\lVert u_V\rVert_2^2}{d'\rho^3}
\right)^2 \ge
\rho^2-
\frac{4\lVert u_V\rVert_2^2}{d'\rho^2} \ge
\rho^2-\frac{4}{d'}\\
&=
\lVert u\rVert_2^2+1-\frac{4}{d'} = 
\lVert y-x\rVert_2^2+1-\frac{4}{d'}.
\end{align*}
Taking square roots proves the lemma.
\end{proof}

We are now ready to prove \cref{lem:simplex-amplification}.
\begin{proof}[Proof of \cref{lem:simplex-amplification}]
Let $Y$ be any random point in $\R^d$. Applying \cref{lem:simplex-avg-dist-amplification} pointwise and taking expectation gives:
\begin{align*}
    \frac{1}{d'+1}
\sum_{i=0}^{d'}
\E
\brk*{
\left\lVert
Y-(x+s_i)
\right\rVert_2
}
&=
\E[\gamma(Y)] \ge
\E
\brk*{
\sqrt{
\lVert Y-x\rVert_2^2
+
1-\frac{4}{d'}
}
}.
\end{align*}

Let $h(t) = \sqrt{ t^2+1-\frac{4}{d'} }$ for $t\ge 0$.
Since $d'>4$, $h''(t) = \frac{
1-\frac{4}{d'}
}{
\left(
t^2+1-\frac{4}{d'}
\right)^{3/2}
}
>0$.
Thus, $h$ is convex, and Jensen's inequality implies that
\begin{align*}
\E
\brk*{
\sqrt{
\lVert Y-x\rVert_2^2
+
1-\frac{4}{d'}
}
}
&=
\E
\brk*{
h\bigl(\lVert Y-x\rVert_2\bigr)
} \ge
h\left(
\E[\lVert Y-x\rVert_2]
\right)\\
&=
\sqrt{
\left(
\E[\lVert Y-x\rVert_2]
\right)^2
+
1-\frac{4}{d'}
}.
\end{align*}

Consequently,
\[
\frac{1}{d'+1}
\sum_{i=0}^{d'}
\E
\brk*{
\left\lVert
Y-(x+s_i)
\right\rVert_2
}
\ge
\sqrt{
\left(
\E[\lVert Y-x\rVert_2]
\right)^2
+
1-\frac{4}{d'}
}.
\]

By averaging, there exists some $i\in\{0,\ldots,d'\}$. Setting $s:=s_i$, we have
\[
\E
\brk*{
\left\lVert
Y-(x+s)
\right\rVert_2
}
\ge
\sqrt{
\left(
\E[\lVert Y-x\rVert_2]
\right)^2
+
1-\frac{4}{d'}
}.
\]
\end{proof}
\paragraph{Proving \cref{thm:Rd-lb-2}.}
\begin{proof}[Proof of \cref{thm:Rd-lb-2}]
Let $d\ge 25$, and let $d' = k = \floor{\sqrt{d}}$. In particular, $d'>4$.

Let $M$ be a strategyproof in expectation mechanism. Let $Y_P := M(P)$ be the (random) point returned by $M$ given input points $P \subset \R^d$.

\paragraph{Recursive instance construction.}
We construct the following instance recursively:
\begin{itemize}
    \item We start with $n = (d' + 1)^k$ agents located at the origin $x_0 := 0$. Let $P_0$ be an instance that consists of $n$ agent locations in $x_0$.
    \item For any $t \in \crl*{0,\ldots,k-1}$: We now describe how to determine $x_{t+1}$ based on $x_t$ and $M$. Let $P_t$ be the instance at the beginning of the step.
    Let $V_t := \crl*{ x \mid x_{l} = 0 \text{ for } l \notin \crl*{td' + 1,\ldots,(t+1)d'}}$ be the $d'$ dimensional Euclidean subspace of coordinates $td' + 1,\ldots,(t+1)d'$.
    For $i \in \crl*{0, \ldots, d'}$:
    \begin{itemize}
        \item We move an agent from $x_t$ to $x_t + s_i^t$ where $S_t :=\crl*{s_0^t, \ldots, s_{d'}^t} \subset V_t$ is a unit regular simplex that spans $V_t$ and $s_i^t \in S_t$ for all $i \in \crl*{0, \ldots, d'}$.
    \end{itemize}
    Repeat this until there is no agent left at $x_t$: all agents at $x_t$ in $P_t$ are equally divided between the unit simplex vertex around $x_t$. Let $P_{t+1}$ be the resulting instance.

    By \cref{lem:simplex-amplification} there exists a point $\hat{s}^t \in S_t$ such that
    \begin{equation}\label{eq:x_t-dist-amplification}
        \E \brk*{ \left\lVert Y_{P_{t+1}} - (x_t + \hat{s}^t) \right\rVert_2 } \ge \sqrt{ \left( \E[\lVert Y_{P_{t+1}} -x_t \rVert_2] \right)^2 + 1-\frac{4}{d'} }.
    \end{equation}

    We choose $x_{t+1} := x_t + \hat{s}^t$.
\end{itemize}

\paragraph{Showing recursive amplified distance.}

Fix a step $t$. 
Let $P_t^0 = P_t, \ldots, P_t^r = P_{t+1}$ be the intermediate instances we obtain after moving a single agent at a time during the $t$ step.
First, by the fact that $M$ is strategyproof in expectation, we get that the expected distance between $M(P_{t}^{i+1})$ and $x_{t}$ is only larger than the one between $M(P_t^{i})$ and $x_{t}$. This is true since we may treat $x_t$ as the true location of the moved agent, and a misreport of the agent to $x_t + s_i^t$ cannot decrease the expected cost of the agent.
The above implies: $D_{P_t^{i+1}}(x_t) \ge D_{P_t^{i}}(x_t)$. Hence, after moving all agents in step $t$ we get:
\begin{equation}\label{eq:dist-increase-due-to-spe}
 D_{P_{t+1}}(x_t) \ge D_{P_t}(x_t)   
\end{equation}
By \cref{eq:x_t-dist-amplification}: $D_{P_{t+1}}(x_{t+1}) \ge \sqrt{(D_{P_{t+1}}(x_t))^2 + 1 - \frac{4}{d'}}$.

Let $\alpha_t := D_{P_t}(x_t)$ for all $t \in \crl*{0,\ldots,k}$. It follows from the definition of $\alpha_t$  that $\alpha_0 \ge 0$.
By the above inequality we get $\alpha_{t+1} \ge \sqrt{\alpha_t^2 + 1 - \frac{4}{d'} }$ which implies:
$\alpha_{t+1}^2 \ge \alpha_t^2 + 1 - \frac{4}{d'}$.
By induction, we get that for any integer $t \ge 1$: $\alpha_{t}^2 \ge t(1 - \frac{4}{d'})$. The base holds trivially, and for the step we assume the claim holds for $t-1$ and thus \[\alpha_t^2 \ge \alpha_{t-1}^2 + 1 - \frac{4}{d'} \ge (t-1)(1-\frac{4}{d'}) + 1 - \frac{4}{d'} = t(1 - \frac{4}{d'}).\]

By denoting the maximum expected cost of $M$ as $\ALG$ we deduce that \begin{equation}\label{eq:Rd-lb-of-2-alg-lb}
    \ALG \ge \alpha_k \ge \sqrt{k\prn*{1 - \frac{4}{d'}}}.
\end{equation}

\paragraph{Bounding the optimum.}
The remaining step is to bound the cost of the optimal solution (and show that it is approximately $\frac{\sqrt{k}}{2}$).
By our construction: $x_t = \sum_{i=0}^{t-1} \hat{s}^i$, $\hat{s}^i \in V_i$ and $\norm{\hat{s}^i}_2 = 1$.
Let $c := \frac{x_{k-1}}{2}$. We now show that choosing $c$ as the facility location yields the desired upper bound on the optimal cost.

First, we bound the distance from $c$ of any child in the recursion tree of $x_{k-1}$ (which is of the form $x_{k-1} + s$ for some $s \in S_{k-1}$):
\begin{align*}
    \norm*{x_{k-1} + s - c}_2^2 & = \norm*{\frac{1}{2}x_{k-1} + s}_2^2 = \frac{1}{4} \norm*{x_{k-1}}^2 + \langle x_{k-1}, s \rangle + \norm*{s}^2 \\
    & \stackrel{(\star)}{=} \frac{1}{4} (k-1) +1 = \frac{k+3}{4},
\end{align*}
where equality $(\star)$ holds for the following reason: First, note that $\langle \hat{s}^i, \hat{s}^j \rangle = 0$ since the two vectors are from orthogonal spaces. Therefore, $\norm{x_{k-1}}^2 = \langle \sum_{i=0}^{k-2} \hat{s}^i, \sum_{i=0}^{k-2} \hat{s}^i  \rangle = \sum_{i=0}^{k-2}\norm{\hat{s}^i}_2^2 = k - 1$, and also $\langle x_{k-1},s\rangle = 0$ for every $s\in S_{k-1}$.

Hence $\norm{x_{k-1}+s-c}_2 \le \frac{\sqrt{k+3}}{2}$ for every $s\in S_{k-1}$.

Next, we consider a child $x_t + s$ for $s \neq \hat{s}^t \in S_t$ of a previous recursion level $t < k-1$. So
\begin{align*}
    \norm*{x_t + s - c}_2^2 & = \norm*{\sum_{i=0}^{t-1} \hat{s}^i + s - \frac{1}{2}\sum_{i=0}^{k-2} \hat{s}^i}_2^2 = \norm*{\frac{1}{2} \sum_{i=0}^{t-1} \hat{s}^i + s - \frac{1}{2}\sum_{i=t}^{k-2} \hat{s}^i}_2^2 \\
    & = \frac{1}{4} \langle  \sum_{i=0}^{t-1} \hat{s}^i,  \sum_{i=0}^{t-1} \hat{s}^i \rangle +  \langle  \sum_{i=0}^{t-1} \hat{s}^i, s \rangle - \frac{1}{2} \langle  \sum_{i=0}^{t-1} \hat{s}^i, \sum_{i=t}^{k-2} \hat{s}^i \rangle + \norm{s}_2^2 \\
    & \quad - \langle s,  \sum_{i=t}^{k-2} \hat{s}^i \rangle + \frac{1}{4} \langle \sum_{i=t}^{k-2} \hat{s}^i, \sum_{i=t}^{k-2} \hat{s}^i \rangle  \\
    & \stackrel{(\star)}{\le} \frac{1}{4} t + 1 + \frac{1}{d'} + \frac{1}{4} (k-2 - t + 1) = \frac{k + 3}{4} + \frac{1}{d'} \stackrel{d'>4}{\le} \frac{k+4}{4},
\end{align*}
where inequality $(\star)$ holds as $\langle \hat{s}^i, \hat{s}^j \rangle = 0$, $\norm{s}_2 = 1$ and $\langle s, \hat{s}^t \rangle = -\frac{1}{d'}$.

Hence $\norm*{x_t + s - c}_2 \le \frac{\sqrt{k+4}}{2}$.

Hence we deduce the following upper bound on the cost of the optimal solution: $\OPT \le \frac{\sqrt{k+4}}{2}$.
Hence, by \cref{eq:Rd-lb-of-2-alg-lb}:
\[
    \frac{\ALG}{\OPT} \ge \frac{\sqrt{k(1 - \frac{4}{d'})}}{\frac{\sqrt{k+4}}{2}} = 2 \sqrt{1 - \frac{4}{d'}} \cdot \sqrt{\frac{k}{k+4}} = 2 - O\prn*{\frac{1}{\sqrt{d}}}.
\]

\end{proof}

\section{Ex-ante objective upper bounds}\label{sec:exantec-ub}

\subsection{Warmup: A tight bound of $1$ for $\R$}

In $\R$ the problem is easy: the following simple mechanism, which chooses the leftmost point w.p. $0.5$ and the rightmost point w.p. $0.5$, is optimal.

\begin{algorithm}[H]
\caption{LR (Left-Right) Mechanism}
\label{alg:lr}
\begin{algorithmic}[1]
\REQUIRE Reported locations $\{x_i\}_{i=1}^n$
\STATE $L \gets \min_{i \in [n]} x_i,\quad R \gets \max_{i \in [n]} x_i$
\STATE Sample $c \sim \mathrm{Unif}(\crl*{L, R})$ \hfill (each endpoint with probability $0.5$)
\RETURN $c$
\end{algorithmic}
\end{algorithm}

\begin{thm}\label{thm:exante-LR-is-sp-and-optimal}
    Given any multiset of agent locations on the real line, \cref{alg:lr} is strategyproof in expectation and has an approximation ratio of $1$.
\end{thm}
\begin{proof}
First, we show that the mechanism is strategyproof in expectation. Fix an agent $i \in [n]$ with true location $x_i$ and fix the reports of all other agents. The case $n=1$ is immediate, so assume $n\ge 2$ and let
\[
    L_{-i}:=\min_{j\ne i}x_j,
    \qquad
    R_{-i}:=\max_{j\ne i}x_j.
\]
If agent $i$ reports $r$, the two endpoints are
\[
    L(r)=\min\{r,L_{-i}\},
    \qquad
    R(r)=\max\{r,R_{-i}\}.
\]
For every report $r$, we have
\[
    |x_i-L(x_i)|\le |x_i-L(r)|
    \qquad\text{and}\qquad
    |x_i-R(x_i)|\le |x_i-R(r)|.
\]
Indeed, each inequality follows immediately by considering whether $x_i$ lies below, inside, or above the interval $[L_{-i},R_{-i}]$. Averaging the two inequalities shows that the expected cost under the truthful report is no larger than under $r$. Hence the mechanism is strategyproof in expectation.

For the approximation ratio, for any $i \in [n]$:
\[
    \E[\abs{x_i - c}] = \frac{1}{2} \abs{x_i - L} + \frac{1}{2} \abs{x_i - R} = \frac{1}{2}(x_i - L) + \frac{1}{2}(R - x_i) = \frac{1}{2}(R-L).
\]

Thus the expected cost of the algorithm is $\ALG := \max_i \E[\abs{x_i - c}] = \frac{1}{2} (R-L)$.

On the other hand $\OPT = \frac{R-L}{2}$ as well (the optimal solution is to return the mid-point of the leftmost and rightmost points) and therefore $\ALG = \OPT$. In particular, the ratio is $1$ whenever $\OPT>0$; when all reports coincide, both costs are $0$.

\end{proof}

\subsection{Upper bounds in $\R^2$}

We start by giving a mechanism that returns a random corner of the minimum bounding box, and then improve the mechanism via introducing random rotations.

\subsubsection{$\rc$ mechanism and the upper bound of $\phi \approx 1.618$}
Consider the following $\rc$ mechanism:

\begin{algorithm}[H]
\caption{$\rc$ Mechanism}
\label{alg:rc}
\begin{algorithmic}[1]
\REQUIRE Reported locations $P = \crl*{(x_i,y_i)}_{i=1}^n \subset \R^2$
\STATE 
\STATE $x_{\min} \gets \min_{i \in [n]} x_i,\quad x_{\max} \gets \max_{i \in [n]} x_i$
\STATE $y_{\min} \gets \min_{i \in [n]} y_i,\quad y_{\max} \gets \max_{i \in [n]} y_i$
\STATE Define the bounding box $B = [x_{\min},x_{\max}] \times [y_{\min},y_{\max}]$
\STATE Let the four-element multiset of indexed corners of $B$ be
\[
C_B \gets \{(x_{\min},y_{\min}), (x_{\min},y_{\max}), (x_{\max},y_{\min}), (x_{\max},y_{\max})\}.
\]
\STATE Sample $z \sim \mathrm{Unif}(C_B)$ \hfill (each indexed corner with probability $0.25$)
\RETURN $z$
\end{algorithmic}
\end{algorithm}

\begin{thm}\label{thm:alg-rc-sp-and-phi-approx}
    \cref{alg:rc} is strategyproof in expectation and has an approximation ratio of exactly $\phi \approx 1.618$.
\end{thm}

The proof of \cref{thm:alg-rc-sp-and-phi-approx} follows from the following couple of lemmas (Lemmas \ref{lem:alg-rc-spe}, \ref{lem:alg-rc-expected-approx-exante-phi}).

\begin{lem}\label{lem:alg-rc-spe}
    \cref{alg:rc} is strategyproof in expectation. 
\end{lem}

\begin{proof}
Fix an agent $i \in [n]$ located at $(x_i,y_i) \in \R^2$, and fix all the locations of the other agents.

The case $n=1$ is immediate, so assume $n\ge 2$. For a coordinate $s\in\R$, let $s^-_{-i}$ and $s^+_{-i}$ denote, respectively, the minimum and maximum reports of the other agents in that coordinate. If agent $i$ reports $r$ in this coordinate, the resulting lower and upper endpoints are
\[
    \ell_s(r):=\min\{r,s^-_{-i}\},
    \qquad
    u_s(r):=\max\{r,s^+_{-i}\}.
\]
For the true coordinate $s_i$ and every report $r$,
\begin{equation}\label{eq:coordinate-endpoints-truthful-closest}
    |s_i-\ell_s(s_i)|\le |s_i-\ell_s(r)|
    \qquad\text{and}\qquad
    |s_i-u_s(s_i)|\le |s_i-u_s(r)|.
\end{equation}
To verify this, consider the three cases $s_i\le s^-_{-i}$, $s_i\in[s^-_{-i},s^+_{-i}]$, and $s_i\ge s^+_{-i}$. In the first case the truthful lower endpoint equals $s_i$ and every possible upper endpoint is at least $s^+_{-i}$; the third case is symmetric. In the middle case, every possible lower endpoint is at most $s^-_{-i}\le s_i$ and every possible upper endpoint is at least $s^+_{-i}\ge s_i$.

Apply \cref{eq:coordinate-endpoints-truthful-closest} separately to the $x$- and $y$-coordinates. Match each corner under a deviation $(x'_i,y'_i)$ with the truthful corner having the same choice of lower or upper endpoint in each coordinate. For every one of these four matched pairs, both absolute coordinate differences from the true point $(x_i,y_i)$ weakly increase under the deviation. Therefore the Euclidean distance to each matched corner weakly increases. Averaging over the four equally likely corners proves strategyproofness in expectation, including deviations that change both coordinates simultaneously.
\end{proof}

\begin{lem}\label{lem:alg-rc-expected-approx-exante-phi}
    \cref{alg:rc} has an expected approximation ratio (w.r.t. $\exantec$) of $\phi \approx 1.618$.
\end{lem}
\begin{proof}
Let $P$ be the given multiset of points (agent locations).
Let $c$ be the center of the minimum enclosing ball of $P$ and let $r$ be the radius of the ball. If $r=0$, all points coincide and the claim is immediate.
W.l.o.g. let us assume that $c = 0$ is the origin (otherwise we may translate the instance) and that $r = 1$ (otherwise we may rescale the instance).
Hence all points are contained in the unit ball and thus the minimum bounding box is contained in the square $[-1,1]^2$.

Let $p = (x,y)$ be some agent location.
Let $B = [a,b] \times [\ell,u]$ be the minimum bounding box of the agent locations, and so
$-1 \le a \le x \le b \le 1$ and $-1 \le \ell \le y \le u \le 1$.
Let $C_B$ denote its four-element multiset of indexed corners.

Let $B' = [-1,1] \times [-1,1]$ be the larger bounding box that contains the box $B$, and let $C' := \crl*{(-1,-1), (-1,1), (1,-1), (1,1)}$ be the set of corners of $B'$.

\[
    \E\brk*{\norm{p - z}_2} = \frac{1}{4} \sum_{z \in C_B} \norm{p-z}_2 \le \frac{1}{4} \sum_{z \in C'} \norm{p - z}_2,
\]
where the inequality follows from the fact that for each indexed corner of $B$ there is a matching corner of $B'$ that is at least as far from $p$.

Since $\norm{p}_2 \le 1$, we obtain the upper bound on the expected cost of the agent located at $p$ by finding an upper bound on $\max_{p: \norm{p}_2 \le 1} G(p)$ where $G(p) := \frac{1}{4} \sum_{z \in C'} \norm{p - z}_2$ (there is such a maximum due to the compactness of $\crl*{p \mid \norm{p}_2 \le 1}$).

Since each summed term in $G$ is convex, $G$ is convex as well. By Bauer maximum principle, $G$ attains its maximum at $\norm{p}_2 = 1$.

So we get that \[\ALG \le \max_{(x,y) \in \R^2 \text{ s.t. } x^2 + y^2=1} \crl*{G(x,y)}.\]

We finish by showing the following claim.
\begin{clm}\label{clm:exante-G-ub-phi}
    \[
        \max_{(x,y) \in \R^2 \text{ s.t. } x^2 + y^2=1} \crl*{G(x,y)} = \phi \approx 1.618.
    \]
\end{clm}
The proof of \cref{clm:exante-G-ub-phi} follows standard two-dimensional function analysis and is deferred to \cref{sec:exante-ub-phi-proof}.

By the claim the expected cost of any agent is at most $\phi$. Thus the maximum over all agent expected costs is at most $\phi$: $\ALG \le \phi$. Under our normalization, $\OPT=1$, which gives the desired upper bound.

To see why the bound is tight, consider the $4$ point instance $P = \crl*{(1,0), (-1,0), (0,1), (0,-1)}$.
The optimal solution is to return the origin and pay the cost $1$. \cref{alg:rc} chooses each of $\crl*{(1,1), (-1,-1), (1,-1), (-1,1)}$ w.p. $0.25$ each. For the agent at $(1,0)$ the expected distance is exactly $G(\cos(0),\sin(0))=\phi$, just as in the analysis above.
\end{proof}

\subsubsection{$\rrc$ mechanism and the upper bound of $1.598$}

The idea in this section is to improve the $\rc$ mechanism by introducing random rotations. For an angle $\theta\in[0,2\pi)$, let
\[
R_\theta
:=
\begin{pmatrix}
\cos\theta & -\sin\theta\\
\sin\theta & \cos\theta
\end{pmatrix}.
\]
Thus $R_\theta$ rotates every column vector counterclockwise by $\theta$,
and $R_\theta^{-1}=R_{-\theta}$. A uniformly random
rotation in $\R^2$ means sampling $\theta\sim\mathrm{Unif}([0,2\pi))$ and
using $R_\theta$.  For a reported profile $P$, write
$R_\theta P:=\crl*{R_\theta p:p\in P}$ for its rotated profile.

\begin{algorithm}[H]
\caption{$\rrc$ Mechanism}
\label{alg:rrc}
\begin{algorithmic}[1]
\REQUIRE Reported locations $P  \subset \R^2$
\STATE Sample $\theta\sim\mathrm{Unif}([0,2\pi))$ and let $R\gets R_\theta$.
\STATE Return $R^{-1} \cdot \rc(RP)$
\end{algorithmic}
\end{algorithm}

The $\rrc$ mechanism randomly rotates the points and then chooses a random corner of the minimum bounding box.

\begin{thm}\label{thm:alg-rrc-sp-and-approx}
    The $\rrc$ mechanism (\cref{alg:rrc}) is strategyproof in expectation and achieves an expected approximation ratio (w.r.t. the $\exantec$ cost) of exactly
    $c = \frac{1}{2\pi}\int_0^{2\pi} \sqrt{3 - 2\sqrt{2}\cos(t)}\,dt \approx 1.5974$.
\end{thm}

\begin{proof}
The strategyproofness in expectation follows directly from the strategyproofness in expectation of \cref{alg:rc} and the fact that the agents have no control over the chosen rotation matrix.

If $\OPT=0$, all reports coincide and the claim is immediate. Just like in the proof of \cref{lem:alg-rc-expected-approx-exante-phi}, we otherwise assume w.l.o.g. that the minimum enclosing ball of the points $P$ is the unit ball. Thus $\OPT = 1$ and, for every $p_i \in P$, $\norm{p_i}_2 \le 1$.

Let $P' = RP$ be the rotated multiset of agent locations.

Since rotation preserves Euclidean distances, $\norm{p'_i}_2 \le 1$ for any $p'_i \in P'$, and still $\OPT = 1$ w.r.t. the rotated points.

We may analyze each agent's expected distance from the returned facility location in the rotated coordinates, since rotating back preserves both the mechanism's distances and the optimal cost.

Let $p = (x,y) \in P$ be some agent location and let $p' = R_{\theta}p \in P'$ be its $\theta$-rotation.

Consider the expected distance between $p'$ and the returned facility location (a random corner). For each random corner of the minimum bounding box there is a corner of the bigger bounding box $[-1,1] \times [-1, 1]$ that is farther from $p'$. 

% Let $z$ be the returned facility location. 

Let
$G(p) := \frac{1}{4} \sum_{z \in C'} \norm{(x,y) - z}_2$.

Given any fixed $\theta$, the expected cost of an agent located at $p$ is at most the expected distance of its rotated point from a random corner in $C'$, which is exactly $G(R_\theta p)$.

By taking expectation over all rotations we get that the agent cost is upper bounded by $h(p) := \E\brk*{G(R_{\theta}p)}$.
By taking the maximum over all points in the unit ball $p: \norm{p}_2 \le 1$ we get the following bound on the ex-ante maximum cost of the algorithm:

\[
    \ALG \le \max_{\norm{p}_2 \le 1} \ \crl*{h(p)}.
\]

Similar to the proof of \cref{lem:alg-rc-expected-approx-exante-phi}: each summed term in $h(p)$ is convex and so $h(p)$ is convex as well. It follows that the maximum over the unit ball occurs on the boundary, which is the unit circle:  $\norm{p}_2 = 1$.
Moreover, $h(R_{\varphi}p)=h(p)$ for every rotation $R_{\varphi}$: this follows by shifting the uniformly distributed angle $\theta$ by $\varphi$ modulo $2\pi$. Hence $h$ is constant on the unit circle, and the maximum is attained at $e_1=(1,0)$.

So: \[
    \ALG \le h((1,0)).
\]

We finish the upper bound proof by showing the following claim:
\begin{clm}\label{clm-h-of-e1}
    \[
        h((1,0)) = \frac{1}{2 \pi} \int_{0}^{2 \pi} \sqrt{3 - 2 \sqrt{2} \cos(t)}\,dt \approx 1.5974.
    \]
\end{clm}
\begin{proof}

Note that $h((1,0)) = \E_{\theta, z \in C'} \brk*{\norm{(1,0) - R_{-\theta} z}_2}$.
For each $z \in C'$, $\norm{z}_2 = \sqrt{2}$ and its angle is in $\crl*{\frac{\pi}{4} + \frac{\pi}{2} \cdot k \mid k \in \crl*{0,1,2,3}}$.

Due to symmetry, each corner angle in $[0,2\pi)$ has the same probability of getting chosen, and thus $R_{-\theta} z$ has the same distribution as $Y$, a uniform random point on the circle around the origin of radius $\sqrt{2}$.

We deduce that 
\begin{align*}
    h((1,0)) &= \E_{t \sim \mathrm{Unif}([0,2\pi))}\brk*{\norm{(1,0) - \sqrt{2} (\cos(t), \sin(t))}_2} \\
    & = \frac{1}{2 \pi} \int_{0}^{2\pi} \sqrt{\prn*{1 - \sqrt{2}\cos(t)}^2 + \prn*{\sqrt{2}\sin(t)}^2}\,dt \\
    & = \frac{1}{2 \pi} \int_{0}^{2\pi} \sqrt{3 - 2\sqrt{2}\cos(t)}\,dt.
\end{align*}
\end{proof}

To see that this expected approximation ratio is tight, let $P_n$ consist of $n\ge3$ evenly spaced points on the unit circle. Then $\OPT(P_n)=1$. As $n\to\infty$, after every rotation the four coordinates of the minimum bounding box converge uniformly to $-1,-1,1,1$: each coordinate-extreme direction is within angle at most $\pi/n$ of a point of $P_n$. Consequently, each returned corner converges uniformly to the corresponding corner of $[-1,1]^2$. By the triangle inequality, replacing a returned corner by its limiting corner changes every agent's distance by at most the corner displacement, which tends uniformly to $0$. The computation in \cref{clm-h-of-e1} therefore shows that the ex-ante cost of every agent in $P_n$ converges to $c$. Hence the supremal approximation ratio of the mechanism is exactly $c$.

\end{proof}

We conclude by remarking that the analogous corner mechanisms do not yield a dimension-independent improvement over $2$ as $d$ grows. As we show in \cref{thm:Rd-lb-2}, for every constant $c>0$, no mechanism family that is strategyproof in expectation can guarantee a $(2-c)$ ex-ante approximation uniformly over all dimensions and numbers of agents.

\section{Lower bounds for $\R^2$}\label{sec:lb-R2}
\label{sec:R2-lower-bounds}

In this section we give lower bounds in $\R^2$ for both the ex-ante and ex-post objectives. The main structure of the construction is similar for both objectives.

\begin{thm}[$\R^2$ lower bounds]
\label{thm:R2-lower-bounds}
Let $M$ be a strategyproof in expectation mechanism in $\R^2$.

Let $a := \frac{1+2\sqrt 7}{3}, \ D := 2\sqrt 3$,
and $\rho := \frac{a}{D} = \frac{1+2\sqrt 7}{6\sqrt 3} \approx 0.6054$, and $\Psi(r) := \frac{1}{2\pi} \int_0^{2\pi} \sqrt{ 1+r^2-2r\cos\theta } \,d\theta$.

The expected approximation ratio of $M$ is at least $1 + \rho \approx 1.605$ for the ex-post objective and at least $\Psi(\rho) \approx 1.09$ for the ex-ante objective.
\end{thm}

\paragraph{High level overview.}
To show \cref{thm:R2-lower-bounds} we give a single hard instance that we use for the lower bounds of both objectives. The instance is constructed in the following way: we start with $\frac{n}{2}$ agents in $u=(-1,0)$ and $\frac{n}{2}$ agents in $v=(1,0)$.
This implies that the expected distance between the point returned by the mechanism $Y$ and at least one of $\crl{u,v}$ is at least $1$. Denote this point by $c$ and the other endpoint by $w_0$. Next, complete $w_0$ to an equilateral triangle $w_0,w_1,w_2$ centered at $c$. The construction moves $\frac{n}{4}$ agents from $c$ to each of $w_1$ and $w_2$. A technical lemma about the average distance from $Y$ to the triangle vertices, together with Jensen's inequality, shows that at least one vertex $b \in \crl*{w_0,w_1,w_2}$ has expected distance at least $a$. Then, $\frac{n}{4}$ agents at $b$ are moved to equally spaced points on the circle $C$ centered at $b$ with radius $D=2\sqrt{3}$. The expected distance between $Y$ and $b$ remains at least $a$. The entire instance is contained in the ball $B(b,D)$, while for every returned point there is a discretization point whose direction is within $\pi/k$ of the antipodal direction.
Thus, the construction moves from two points, to an equilateral triangle, to a circle (\cref{fig:r2-construction}).

\paragraph{Expected distance amplification via an equilateral triangle.}
We first prove the geometric amplification inequality that we use
in the construction.

Let $c=(0,0)$, 
and let $
w_0=(2,0), 
w_1=(-1,\sqrt 3),
w_2=(-1,-\sqrt 3)$.
Thus, $(w_0,w_1,w_2)$ are the vertices of an equilateral triangle
centered at $c$, with circumradius $2$ and side length $\norm{w_i-w_j}_2=2\sqrt 3=D$ for every $i\neq j$.
Let
$
g(r)
:=
\frac{
\abs{r-2}
+
2\sqrt{r^2+2r+4}
}{3},
\qquad
r\ge 0.
$
The following technical lemma shows a lower bound on the average distance of the triangle points from any point in the plane.
\begin{lem}[Equilateral-triangle distance inequality]
\label{lem:equilateral-pointwise}
For every $y\in\mathbb R^2$,
$
\frac13
\sum_{j=0}^2
\norm{y-w_j}_2
\ge
g\prn*{\norm{y-c}_2}.
$
Moreover, $g$ is convex and nondecreasing on $[0,\infty)$.
\end{lem}
The proof of \cref{lem:equilateral-pointwise} consists of standard calculus and is deferred to \cref{sec:missing-analysis-lb-R2}. A direct corollary of \cref{lem:equilateral-pointwise} gives a lower bound on the expected distance from any randomly chosen point $Y$.

\begin{cor}[Expected triangle amplification]
\label{cor:equilateral-expected}
Let $c\in\mathbb R^2$, and let $w_0,w_1,w_2$ be the vertices of an equilateral triangle centered at $c$, with
circumradius $2$. For every random point $Y$, there exists $j \in \crl*{0,1,2}$ such that 
\[
\E\brk*{
\norm{Y-w_j}_2
}
\ge
g\prn*{
\mathbb E\brk*{
\norm{Y-c}_2
}
}.
\]
\end{cor}

\begin{proof}
By \cref{lem:equilateral-pointwise}, pointwise, $\frac13
\sum_{j=0}^2
\norm{Y-w_j}_2
\ge
g\prn*{\norm{Y-c}_2}$.
Taking expectation and applying Jensen's inequality ($g$ is convex) gives
\[
\frac13
\sum_{j=0}^2
\mathbb E\brk*{
\norm{Y-w_j}_2
}
\ge
\mathbb E\brk*{
g\prn*{\norm{Y-c}_2}
} \ge g\prn*{
\mathbb E\brk*{
\norm{Y-c}_2
}
},
\]
which implies the desired.
\end{proof}

\iffalse
If
\[
\mathbb E\brk*{
\norm{Y-c}_2
}
\ge1,
\]
then the monotonicity of $g$ gives
\[
\frac13
\sum_{j=0}^2
\mathbb E\brk*{
\norm{Y-w_j}_2
}
\ge
g(1).
\]
Finally,
\[
g(1)
=
\frac{1+2\sqrt7}{3}
=
a.
\]
The conclusion follows by averaging.
\fi

\paragraph{Proof of \cref{thm:R2-lower-bounds}.}
We can now construct an instance we use for both objectives.
Let $k \gg 1$ and $n = 4k$.

\begin{lem}[Common $\R^2$ hard instance]
\label{lem:common-planar-hard-instance}
Let $M$ be a strategyproof in expectation mechanism in $\mathbb R^2$, and $k \in \N, k > 5$.

Then there exists an instance $P$ and $b \in \R^2$ such that at least one agent is located at each of $z_0,\ldots,z_{k-1}$, $D_P(b) \ge a$ and $\OPT(P) \le D$ where $z_0,\ldots,z_{k-1} \in \R^2$ are the discretization points of the circle centered at $b$ with radius $D$ via $k$ points. That is, for $j\in\{0,\ldots,k-1\}$,  $z_j := b+ D \prn*{ \cos\prn*{\frac{2\pi j}{k}}, \sin\prn*{\frac{2\pi j}{k}}}$.
\end{lem}

\begin{proof}
Let $u = (0,0)$, $v = (2,0)$. So $\norm{u-v}_2 = 2$.
Let $P^0$ be an instance in which $2k$ agents are located at
$u$, and $2k$ agents are located at $v$.

For every realization $y$ of $Y_{P^0}$, triangle inequality
gives
$
\norm{y-u}_2+\norm{y-v}_2
\ge
\norm{u-v}_2
=
2.
$
By taking expectation we get $D_{P^0}(u)+D_{P^0}(v)\ge 2$.
Therefore, one of $u,v$ has expected distance at least $1$. Denote
this endpoint by $c$, and denote the other endpoint by $w_0$. Thus, $D_{P^0}(c)\ge1$
and $\norm{w_0-c}_2=2$.
Complete $w_0$ to an equilateral triangle $w_0,w_1,w_2$
centered at $c$, with circumradius $2$. Hence, for every $i\neq j$,
$\norm{w_i-w_j}_2=D=2\sqrt3$.

There are at least $2k$ agents at $c$. One by one, we
move $k$ agents from $c$ to $w_1$, and move another $k$ agents from $c$ to $w_2$. Let $P^1$
be the resulting profile.
Since $M$ is strategyproof in expectation, we have that $D_{P^1}(c)
\ge
D_{P^0}(c)
\ge 1$ since we move the agents one by one and have that the expected distance between $Y$ and $c$ may only increase as a result of the change since otherwise an agent at $c$ could misreport its location to $w_1$ or $w_2$ and improve her expected cost.

Applying \cref{cor:equilateral-expected} to $Y_{P^1}$, there exists
a vertex $b\in\{w_0,w_1,w_2\}$ such that \[D_{P^1}(b)\ge g(D_{P^1}(c)) \ge g(1) \ge a.\]

Every triangle vertex contains at least $k$ agents: $w_1$ and
$w_2$ each contain $k$ agents, while $w_0$ contains at least
$2k$ agents.

For $j\in\{0,\ldots,k-1\}$, let $z_j := b+ D \prn*{ \cos\prn*{\frac{2\pi j}{k}}, \sin\prn*{\frac{2\pi j}{k}}}$ be the $j$'th point of the discretization of the circle centered at $b$ with radius $D$ via $k$ points.
Select $k$ agents located at $b$, and move one selected agent to
each of the points $z_0,\ldots,z_{k-1}$.
Let $P$ be the resulting profile. Again, since $M$ is strategyproof in expectation, no agent at $b$ may benefit from misreporting to any of the points $z_j$, and hence:
$D_P(b) \ge D_{P^1}(b) \ge a$.

Every point $z_j$ is at distance exactly $D$ from $b$. Every
triangle vertex is at distance at most $D$ from $b$, and the
triangle center $c$, if it remains occupied, satisfies
$\norm{c-b}_2=2<D$.
Therefore, every occupied point of $P$ lies in $B(b,D)$, and hence
$\OPT(P)\le D$.
\end{proof}

\paragraph{The ex-ante lower bound.}

For $D>0$, let $\Gamma_D(r): \R_{\ge 0} \to \R_{\ge 0}$ be defined as \[\Gamma_D(r) := \frac1{2\pi} \int_0^{2\pi} \sqrt{ r^2+D^2-2rD\cos\theta } \,d\theta.\]
$\Gamma_D(r)$ is the average distance from a point at distance $r$ from the center of a circle of radius $D$ to a uniformly random point on that circle.
By scaling,
$\Gamma_D(r) = D\Psi\prn*{\frac rD}$. Next, we give a useful lemma; in the first item we show that $\Gamma_D$ is convex and nondecreasing, which will be useful for the ex-ante lower bound. In the second item we bound the discretization error of the circle average by a negligible $O(\frac{1}{k})$ term.

\begin{lem}
\label{lem:circle-average-and-discretization} The following two items hold:
\begin{enumerate}[label=(\roman*)]
\item \label{lem:circle-average-properties}
\emph{(Circle average properties.)}
For every $D>0$, the function $\Gamma_D$ is convex and nondecreasing.
\item \label{lem:circle-discretization-error}
\emph{(Circle discretization error.)}
Let $b,y\in\mathbb R^2$, let $D>0$, and for any $j \in \crl*{0,\ldots,k-1}$ let $z_j = b+ D \prn*{ \cos\prn*{\frac{2\pi j}{k}}, \sin\prn*{\frac{2\pi j}{k}}}$.
Then
\[
\frac1k
\sum_{j=0}^{k-1}
\norm{y-z_j}_2
\ge
\Gamma_D(\norm{y-b}_2)
-
\frac{\pi D}{2k}.
\]
\end{enumerate}
\end{lem}

The proof of \cref{lem:circle-average-and-discretization}, which uses standard technical calculus, is deferred to \cref{sec:missing-analysis-lb-R2}.

\begin{lem}[Ex-ante lower bound]
\label{lem:R2-ex-ante-lb}
Let $P,b,z_0,\ldots,z_{k-1}$ be the profile and points obtained in
\cref{lem:common-planar-hard-instance}. Then
\[
\frac{\exantec(P,M)}{\OPT(P)}
\ge
\Psi(\rho)-\frac{\pi}{2k}.
\]
\end{lem}

\begin{proof}

Because every $z_j$ is occupied by an agent, \[\exantec(P,M) \ge \frac1k \sum_{j=0}^{k-1} \E \brk*{\norm{Y_P-z_j}_2}= \E\brk*{ \frac1k \sum_{j=0}^{k-1} \norm{Y_P-z_j}_2 }, \] where the last equality is by linearity of expectation.
Applying \cref{lem:circle-average-and-discretization}(ii) pointwise to $Y_P$
gives
\[
\exantec(P,M)
\ge
\mathbb E\brk*{
\Gamma_D(\norm{Y_P-b}_2)
}
-
\frac{\pi D}{2k}.
\]
By the convexity of $\Gamma_D$ (\cref{lem:circle-average-and-discretization}(i)) and Jensen's inequality,
\[
\mathbb E\brk*{
\Gamma_D(\norm{Y_P-b}_2)
}
\ge
\Gamma_D\prn*{
\mathbb E\brk*{
\norm{Y_P-b}_2
}
}.
\]
By \cref{lem:common-planar-hard-instance},
$\mathbb E\brk*{\norm{Y_P-b}_2 }=D_P(b)\ge a.
$
Since $\Gamma_D$ is nondecreasing (\cref{lem:circle-average-and-discretization}(i)), we have
$\exantec(P,M) \ge \Gamma_D(a)-\frac{\pi D}{2k}$.
Using
$
\Gamma_D(a)
=
D\Psi\prn*{\frac aD}
=
D\Psi(\rho)
$
and
$
\OPT(P) \le D,
$
we obtain
\[ \frac{\exantec(P,M)}{\OPT(P)} \ge \Psi(\rho)-\frac{\pi}{2k}. \]
\end{proof}

\paragraph{The ex-post lower bound.}
We now use the same profile but exploit the fact that, for every
realized mechanism output, one of the discretized circle vertices is
almost directly opposite (anti-podal) to the output $Y$ relative to the center $b$.

Let $H_{D,k}(r) := \sqrt{ r^2+D^2+ 2rD\cos\prn*{\frac{\pi}{k}}}$.

\begin{lem}
\label{lem:farthest-vertex-and-radial-convexity}
The following two items hold: \begin{enumerate}[label=(\roman*)]
\item \label{lem:farthest-discretized-circle-vertex}
\emph{(Farthest discretized circle vertex.)}
For every $y\in\mathbb R^2$, and $r:=\norm{y-b}_2$ we have
\[
\max_{0\le j\le k-1}
\norm{y-z_j}_2
\ge
\sqrt{
r^2+D^2+
2rD\cos\prn*{\frac{\pi}{k}}
}.
\]
\item \label{lem:ex-post-radial-convexity}
$H_{D,k}$ is convex and nondecreasing on $[0,\infty)$.
\end{enumerate}
\end{lem}

The proof of \cref{lem:farthest-vertex-and-radial-convexity} is standard and deferred to \cref{sec:missing-analysis-lb-R2}.

\begin{lem}[Ex-post lower bound]
\label{lem:R2-ex-post-lb}
Let $P,b,z_0,\ldots,z_{k-1}$ be the profile and points obtained in
\cref{lem:common-planar-hard-instance}. Then
\[
\frac{\expostc(P,M)}{\OPT(P)}
\ge
\sqrt{1+\rho^2+2\rho\cos\prn*{\frac{\pi}{k}}}.
\]
\end{lem}

\begin{proof}

Since the points $z_0,\ldots,z_{k-1}$ contain at least one agent each, we have
\[
\expostc(P,M) \ge \mathbb E\brk*{ \max_{0\le j\le k-1} \norm{Y_P-z_j}_2 }. \]
By \cref{lem:farthest-vertex-and-radial-convexity}(i), $\expostc(P,M) \ge \mathbb E\brk*{ H_{D,k}(\norm{Y_P-b}_2) }$.
By the convexity of $H_{D,k}$ (\cref{lem:farthest-vertex-and-radial-convexity}(ii)) and Jensen's inequality, $\expostc(P,M) \ge H_{D,k}\prn*{\mathbb E\brk*{\norm{Y_P-b}_2}}$.
By \cref{lem:common-planar-hard-instance},
$\mathbb E\brk*{\norm{Y_P-b}_2}=D_P(b)\ge a$.
Since $H_{D,k}$ is nondecreasing,
$
\expostc(P,M)
\ge
H_{D,k}(a).
$
Therefore,
\[
\expostc(P,M)
\ge
\sqrt{
a^2+D^2+
2aD\cos\prn*{\frac{\pi}{k}}
}.
\]
Dividing by
$
\OPT(P) \le D
$
and using
$
\rho=\frac aD
$
gives
\begin{align*}
    \frac{\expostc(P,M)}{\OPT(P)} & \ge \sqrt{1+\rho^2+2\rho\cos\prn*{\frac{\pi}{k}}}.
\end{align*}
\end{proof}

By \cref{lem:R2-ex-post-lb}, \cref{lem:R2-ex-ante-lb} and taking $k \to \infty$ we get the desired.

\section{Best of both worlds: The egalitarian approximation of $\rrcwm$}\label{sec:egalitarian-approx-of-rrcwm}
$\rrcwm$ is the currently best known mechanism for the utilitarian objective in $\R^d$ for any $d \in \N$ \cite{barak2026facilitylocationmechanismdesign}, achieving an expected approximation ratio in $[\sqrt{2} - o_d(1), 1.55]$. In this section, we show that it is also not a bad mechanism for the egalitarian objective, achieving an approximation ratio of $2$ for the ex-post objective in $\R^d$ for any $d \in \N$.

\begin{thm}\label{thm:rrcwm-max-cost}
For every $d\geq1$, the tight egalitarian approximation ratio of $\rrcwm$ is $2$ in $\R^d$. The $2$ approximation bound holds for both the ex-ante and ex-post objectives.
\end{thm}

We also show that, perhaps surprisingly, the deterministic version of $\cwmed$ is not a $2$-approximation for the ex-post egalitarian objective in $\R^d$. It has an approximation ratio of $1 + \sqrt{2} \approx 2.41$, which is worse than a deterministic dictator mechanism, which achieves a ratio of $2$.

\paragraph{The $\cwmed$ and $\rrcwm$ mechanisms.}
Let us first give a proper definition of the mechanism.
First, we define the deterministic $\cwmed$ mechanism, which is the coordinate-wise median mechanism.
The coordinate-wise median mechanism is defined as follows.
\begin{definition}[$\cwmed$ - Coordinate-Wise Median]\label{def:cwmed}
    For a multiset of $n$ points $P \subseteq \R^d$:
    $\cwmed(P) := (l_1, \ldots, l_d)$ where $l_j$ is the median of $P_j$ for all $j \in [d]$.
\end{definition}

If $n$ is even ($n=2k$), we assume that,
conditional on the realized rotation, the $k$-th and $(k+1)$-st
order statistics are selected with probability $1/2$ each independently in every
coordinate.

To define the $\rrcwm$ mechanism, we first define the group of rotation matrices in $\R^d$.
\begin{definition}[Rotation Matrix]\label{def:rotation-matrix}
    A \emph{rotation matrix} in $\mathbb{R}^d$ is an orthogonal matrix $R \in \mathbb{R}^{d \times d}$ with determinant $\det(R)=1$. Equivalently,
    \[
    R^\top R = I_d \quad \text{and} \quad \det(R)=1,
    \]
    where $I_d$ is the $d \times d$ identity matrix. The set of all such matrices forms the \emph{special orthogonal group}
    \[
    SO(d) := \{ R \in \mathbb{R}^{d \times d} \;:\; R^\top R = I_d,\ \det(R)=1 \}.
    \]

    In $\R^2$, the two-dimensional \emph{rotation matrix} by angle $\theta \in [0,2\pi)$ is defined as
    \[
        R_\theta \;=\;
        \begin{pmatrix}
            \cos\theta & -\sin\theta \\
            \sin\theta & \phantom{-}\cos\theta
        \end{pmatrix}.
    \]
\end{definition}

The mechanism $\rrcwm$ of \cite{barak2026facilitylocationmechanismdesign} is defined as follows: given an instance $P = (x_1, \ldots, x_n)$, it samples a random rotation matrix $R \in SO(d)$, applies the coordinate-wise median mechanism to the rotated instance $R(P) = (R(x_1), \ldots, R(x_n))$, and then rotates back the output. Formally, let $\cwmed$ denote the coordinate-wise median mechanism. 
Then:
\begin{algorithm}[H]
\caption{$\rrcwm(P)$: Randomly Rotated Coordinate-wise Median}
\label{alg:rrcwmed}
\begin{algorithmic}[1]
\REQUIRE $P = \{p_1,\dots,p_n\} \subseteq \mathbb{R}^d$
\STATE Sample a uniformly random rotation matrix $R \in SO(d)$. \COMMENT{Haar uniform\footnotemark}
\STATE Compute the rotated dataset $P' \gets \{p'_i = R p_i : p_i \in P\}$.
\STATE Compute the coordinate-wise median $h' \gets \cwmed(P')$.
\RETURN $h \gets R^{-1} h'$. \COMMENT{$R^{-1}=R^\top$ for $R\in SO(d)$}
\end{algorithmic}
\end{algorithm}
\footnotetext{A random matrix $R \in SO(d)$ is said to be Haar-uniform if its distribution is the unique probability measure on $SO(d)$ that is invariant under left and right multiplication by any fixed rotation \cite{Vershynin2025HDP,mezzadri2006generate}.}

\subsection{The egalitarian approximation of $\cwmed$}

We now show that $\cwmed$ has a tight dimension-independent approximation ratio of $1 + \sqrt{2}$ for the egalitarian objective.
\begin{thm}\label{thm:cwmed-egalitarian}
The $\cwmed$ mechanism has an approximation ratio of at most $1+\sqrt{2}$ for the egalitarian objective in $\R^d$ for every $d \in \N$. Furthermore, this bound is tight up to $o_d(1)$.
\end{thm}

The upper bound of \cref{thm:cwmed-egalitarian} follows directly from \cite{Hastings26} (Theorem 3). We include the upper bound here for completeness.

\begin{proof}
Fix an instance $P=(p_1,\ldots,p_n)$ and let $m=\cwmed(P)$ be the returned facility location. Let $c^*$ be a center of a minimum enclosing ball of $P$, and let $r^*$ be its radius. Translate the profile and its coordinate-wise median by $-c^*$. We may therefore assume that $c^*=0$, so $\norm{p_i}_2\le r^*$ for all $i\in[n]$. 

Fix $j\in[d]$. If $m_j>0$, at least $n/2$ of the reported $j$-th
coordinates are at least $m_j$, and the corresponding reports satisfy
$p_{ij}^2\geq m_j^2$. If $m_j<0$, the same argument applies to the
coordinates at most $m_j$; if $m_j=0$, the resulting inequality is
immediate. Hence $\sum_{i=1}^n p_{ij}^2\geq\frac{n}{2}m_j^2$.

By summing over all coordinates we get
\[
\frac{n}{2}\norm{m}_2^2 \le \sum_{i=1}^n\norm{p_i}_2^2 \le n(r^*)^2,
\]
and therefore $\norm{m}_2\le\sqrt{2}r^*$. By triangle inequality
we get that for every $i\in[n]$:
\[
\norm{p_i-m}_2 \le \norm{p_i}_2+\norm{m}_2 \le (1+\sqrt{2})r^*.
\]
Taking the maximum over $i$ proves the upper bound.

To prove tightness, fix $k\geq1$, let $d=2k$. Let $u:=\frac{1}{\sqrt{2k}} (1, \ldots, 1)$, $u^-:= -u$, $a:=\frac{1}{\sqrt{k+1}}$. For each $t \in \crl*{0, \ldots, 2k-1}$, let $S_t:=\{t,t+1,\ldots,t+k\}\pmod{2k}$ be a set of $k+1$ coordinates starting at $t$ (mod $d$), and $p_t:= a \cdot \mathbf{1}_{S_t}$. That is, $p_t$ has value $a$ in coordinates $S_t$ and $0$ in the other coordinates.
Consider the instance $P_k:=\prn*{u^-, (p_t)_{t \in \crl*{0, \ldots, 2k-1}}}$.
$P_k$ has $2k+1$ reports, all of unit norm. Each coordinate belongs to
exactly $k+1$ of the sets $S_t$, so its reported values consist of one
negative value, $k-1$ zeros, and $k+1$ copies of $a$. Hence, the median in each axis is $a$, and thus the coordinate-wise median of $P_k$ is exactly at $\cwmed(P_k) = (a,\ldots,a)$.
Since the distance between $u^-$ to $\cwmed(P_k)$ is
\begin{align*}
    \norm{u^- - \cwmed(P_k)}_2 & = \sqrt{\sum_{j=1}^d (a - u^-_j)^2} = \sqrt{2k} \cdot \prn*{\frac{1}{\sqrt{k+1}} + \frac{1}{\sqrt{2k}}} \\
    & = 1 + \frac{\sqrt{2k}}{\sqrt{k+1}},
\end{align*}
we have that $\ALG \ge 1 + \sqrt{\frac{2k}{k+1}}$. The unit ball contains all of $P_k$ and so $OPT \le 1$. Therefore, the approximation ratio of $\cwmed$ on $P_k$ is at least $1 + \sqrt{\frac{2k}{k+1}}$, which tends to $1 + \sqrt{2}$ as $k \to \infty$.
\end{proof}

The dimension-independent constant is not the exact ratio in every
fixed dimension. In the two lowest dimensions the exact factor is $2$ \citep{goel2023optimality}.

\subsection{The egalitarian approximation of $\rrcwm$}

Before showing \cref{thm:rrcwm-max-cost}, we first prove the following lemma, which states that given $s$ vectors and a random direction, the expected squared median of the projections of the vectors onto the random direction is at most $1/d$ times the squared radius of the smallest ball containing all the vectors.

\begin{lem}\label{lem:directional-median}
Let $q_1,\ldots,q_s\in\mathbb{R}^d$, where $s$ is odd, and assume
that $\forall i \in [s]$: $\norm{q_i}_2\le r$.
For $g\in\mathbb{R}^d$, let $F(g):=\med{i\in[s]}\langle g,q_i\rangle$ denote the median of the projections of $q_i$ onto $g$.
If $U$ is uniformly distributed on the unit sphere
$\mathbb{S}^{d-1}$, then $\E \brk*{F(U)^2} \le \frac{r^2}{d}$.
\end{lem}

\begin{proof}
The proof of $d=1$ is trivial, so we assume $d\ge 2$.
We start by stating a claim from \cite{barak2026facilitylocationmechanismdesign} that we will use in the proof. See Claim~26 of \cite{barak2026facilitylocationmechanismdesign}.

\begin{clm}\label{clm:medians-diff-bund}
    Let $n = 2k+1 \in \N$ be an odd natural number. Let $\med(w)$ be the median of $w = (w_1,\ldots,w_n) \in \R^n$. Then for any $x,y \in \R^n$: \[
    \abs{\med(x) - \med(y)} \le \norm{x - y}_{\infty}.
    \]
\end{clm}

\cref{clm:medians-diff-bund} states that the median is $1$-Lipschitz with respect to the $\ell_\infty$ norm. Applying this inequality to
$x$ and $y$ defined by $x_i=\langle g,q_i\rangle$, $y_i=\langle g',q_i\rangle$ for $i \in [s]$ yields 
\[
|F(g)-F(g')| \le \max_{i\in[s]} \left|\langle g-g',q_i\rangle\right| \le r\norm{g-g'}_2.
\]
Thus $F$ is $r$-Lipschitz. Let $G\sim N(0,I_d)$. By the Gaussian Poincar\'e inequality and its extension to Lipschitz functions, $\Var (f(G))\le L^2$ for every $L$-Lipschitz function $f$; see Theorem~3.20 and the paragraph immediately following it of \cite{boucheron2013concentration}. Consequently, $\Var(F(G)) \le r^2$.
Because $s$ is odd, negating all projected values negates their
median, and hence $F(-g)=-F(g)$.
Since $G$ and $-G$ have the same distribution,
$\E \brk*{F(G)}=0$. Therefore (since $\Var(X)=\E \brk*{X^2}-(\E \brk{X})^2$ for any random variable $X$),
\begin{equation}\label{eq:directional-median-variance}
    \E \brk*{F(G)^2} =\operatorname{Var}(F(G)) \le r^2.
\end{equation}

Next, we may write $G$ as $G=\rho U$, where $\rho=\norm{G}_2$, $U$ is uniform on
$\mathbb{S}^{d-1}$, and $\rho$ and $U$ are independent (see \cite{Vershynin2025HDP}, Sec.~3.3.3, Eq.~(3.15), p.~69, and
Exercise~3.22, p.~93).
Positive homogeneity of the median gives $F(\rho U)=\rho F(U)$.
Moreover, \[ \E \brk*{\rho^2}= \E\brk*{\norm{G}_2^2} = \sum_{i=1}^d \E\brk*{G_i^2} = \sum_{i=1}^d (Var(G_i) + (\E[G_i])^2) = \sum_{i=1}^d (1 + 0) = d.\] It follows that
\[
    d \ \E \brk*{F(U)^2} = \E\brk*{\rho^2} \cdot \E \brk*{F(U)^2} = \E \brk*{\rho^2 F(U)^2} = \E \brk*{F(G)^2} \le r^2,
\]
where the second equality follows from independence of $\rho$ and $U$, and the last inequality follows from \cref{eq:directional-median-variance}.
This concludes the proof of the lemma.
\end{proof}

We are now ready to give the proof of \cref{thm:rrcwm-max-cost}.
\begin{proof}[Proof of Theorem~\ref{thm:rrcwm-max-cost}]
We show the bound for the ex-post objective, which implies the same bound for the ex-ante objective as well. 
Let $c^*$ be the center of a minimum enclosing ball of $P$, and let
$r^*$ be its radius. $\rrcwm$ is translation equivariant, so we may translate the profile and assume without loss of generality that $c^*=0$.
Therefore, $\norm{p_i}_2 \le r^*$ for every $i \in [n]$.
If $d=1$, every possible median lies in $[-r^*,r^*]$, and the result follows directly from the triangle inequality. Hence, assume $d\geq2$.
First suppose that $n$ is odd. For the sampled
rotation matrix $R$, let $u_j$ be the $j$'th row of $R$. Since $R$ is
Haar-uniform in $SO(d)$, every $u_j$ is marginally uniform on
$\mathbb{S}^{d-1}$ \cite{Vershynin2025HDP}.

Let $m$ be the coordinate-wise median of the rotated instance $RP$.
Let $F(g):=\med{i\in[n]} \langle g,p_i \rangle$. We have $m_j =\med{i \in [n]} (Rp_i)_j = \med{i \in [n]} \langle u_j,p_i \rangle =F(u_j)$.
Lemma~\ref{lem:directional-median} therefore implies $\E \brk*{m_j^2} \le \frac{(r^*)^2}{d}$ for all $j \in [d]$.
We note that the independence of the rows is not required. Summing the above over all $j \in [d]$ yields
\[
\E \brk*{\norm{m}_2^2} =\sum_{j=1}^d\E \brk*{m_j^2} \le (r^*)^2.
\]

Let $h = R^\top m$ denote the output in the original coordinates (the output of the mechanism). Since $\norm{\cdot}_2$ is rotation invariant,
\[
\E \brk*{\norm{h}_2} \le \sqrt{\E \brk*{\norm{h}_2^2}} =\sqrt{\E \brk*{\norm{m}_2^2}} \le r^*.
\]
For every realization of $h$, the triangle inequality gives
\[
\max_i\norm{p_i-h}_2 \le \max_i\bigl(\norm{p_i}_2+\norm{h}_2\bigr) \le r^*+\norm{h}_2.
\]
Taking expectations gives $\expostc = \E \brk*{\max_i\norm{p_i-h}_2} \le 2 r^*$, which proves the upper bound for odd $n$.

To handle the case of even $n$, suppose that $n=2k$, and let $m$ again denote the coordinate-wise median of $RP$. Independently of $R$ and of one another, choose an index $I_j$ uniformly from $[n]$ for each coordinate $j$. Conditional on $R$, we may couple the mechanism's median choice in coordinate $j$ with $I_j$ so that $m_j$ is the unique median after deleting report $I_j$: deleting one of the first $k$ ranked observations leaves the $(k+1)$-st order statistic, while deleting one of the last $k$ leaves the $k$-th. This remains valid with ties after fixing an arbitrary ordering among tied indexed observations. For each fixed $I_j=i$, Lemma~\ref{lem:directional-median}, applied to the $n-1$ remaining reports, gives $\E\brk*{m_j^2\mid I_j=i}\leq (r^*)^2/d$. Averaging over $I_j$ and summing over the coordinates yields $\E\brk*{\norm{m}_2^2}\leq(r^*)^2$, so the rest of the odd-$n$ argument applies unchanged.

To see that the bound is tight, one may consider the profile consisting of two agents at the origin and one agent at $e_1 = (1,0,\ldots,0)$. The minimum enclosing ball is centered at $e_1/2$ and has radius $1/2$. For every rotation, the coordinate-wise median is the origin, so $\rrcwm$ returns the origin deterministically. Its maximum cost is $\norm{e_1}_2 = 1$, and thus the approximation ratio is exactly $2$.
\end{proof}

\section{Discussion}\label{sec:discussion}

In this paper we study the egalitarian facility location mechanism design problem for both the widely studied ex-post objective and the ex-ante objective which we introduce for this problem (which makes sense for strategyproof in expectation mechanisms that already take an agent's cost to be her expected cost).

Many natural mechanism ideas for the ex-post objective for $\R^2$ exist, but many of these fail to get a strictly better constant than the trivial $2$ deterministic approximation. In retrospect, our lower bound of $2 - o_d(1)$ can be used to quickly reject many of the mechanisms, those that may be naturally extended to $\R^d$ and get a $2 - \Omega(1)$ approximation (since if these turned out to be strategyproof in expectation, it would contradict our high-dimensional lower bound). Some of these attempts include variations of following a distribution over the support of the minimum enclosing sphere and its center, or randomly projecting the points onto a random direction and using (a variation of) the real line mechanism of \cite{procaccia2013approximate}. In other words, if there is a $2 - \Omega(1)$ approximate mechanism which is strategyproof in expectation in $\R^2$, it must crucially rely on the fact that all points lie in low-dimensional space.

It remains open to close the remaining gaps for both the ex-ante and ex-post objectives in $\R^2$.

A different future direction would be to study multi-objective optimization and to study the tradeoff between the different objectives.

Another exciting future direction is to study weaker egalitarian objectives (such as minimizing a percentile cost rather than a maximum cost) and discover whether breaking the deterministic dictator barrier is possible for other objectives.

A final future direction would be to consider weaker notions of strategyproofness, such as approximate strategyproofness, and study whether this allows circumventing the lower bounds we show.
\section*{Acknowledgments}
The work of Z.\ Barak and I.\ Talgam-Cohen was supported by the European Research Council (ERC) under the European Union’s Horizon 2020 research and innovation program (grant agreement No.~101077862, project ALGOCONTRACT), by the Israel Science Foundation (grant No.~3331/24), by the NSF-BSF (grant No.~2021680), and by a Google Research Scholar Award.

\paragraph{AI disclosure.}
The authors have used ChatGPT (versions 5.4 to 5.6), Claude Opus (versions 4.8, 5) and Gemini (version 3.1) for editing, proofreading, simplifying, rephrasing parts of some arguments and writing the code to create the figures in the paper. Also, the authors have used ChatGPT and Gemini for generating the initial draft proofs specifically of the standard simplex identities lemma (\cref{lem:simplex-identities}) and the standard undergrad calculus technical lemmas: \cref{lem:equilateral-pointwise}, \cref{lem:circle-average-and-discretization}, \cref{lem:farthest-vertex-and-radial-convexity}, which the authors further manually refined, simplified and clarified.
The authors have verified all AI outputs, substantially revised them, and take full responsibility for their correctness.

\appendix
%\section*{Organization of Appendices}
%\cref{sec:exante-ub-phi-proof} contains missing proofs from \cref{sec:exantec-ub}. \cref{sec:missing-analysis-lb-2-Rd} contains missing proofs from \cref{sec:lb-2-Rd}, and \cref{sec:missing-analysis-lb-R2} contains missing proofs from \cref{sec:lb-R2}.

\section{Proofs from \cref{sec:exantec-ub}}\label{sec:exante-ub-phi-proof}

\subsection{Proof of Claim \ref{clm:exante-G-ub-phi}}

\begin{proof}
By symmetry of the four corners, $G(x,y)$ depends only on $|x|, |y|$ and thus we may restrict to $x = \cos(t)$, $y = \sin(t)$, $t \in [0, \frac{\pi}{4}]$. Then:

\begin{align*}
    G(x,y) & = \frac{1}{4}\sum_{z \in C'} \norm{(x,y)-z}_2 \\
    & = \frac{1}{4}\Bigl(\norm{(x,y) - (1,1)}_2 + \norm{(x,y) - (-1,1)}_2 \\
    &\hspace{4.1em}{}+ \norm{(x,y) - (1,-1)}_2 + \norm{(x,y) - (-1,-1)}_2\Bigr).
\end{align*}

Since
\begin{align*}
    & \norm{(x,y) - (1,1)}_2^2 = (\cos(t) - 1)^2 + (\sin(t) - 1)^2 = 3 - 2(\cos(t) + \sin(t)), \\
    & \norm{(x,y) - (-1,1)}_2^2 = (\cos(t) + 1)^2 + (\sin(t) - 1)^2 = 3 + 2(\cos(t) - \sin(t)), \\
    & \norm{(x,y) - (1,-1)}_2^2 = (\cos(t) - 1)^2 + (\sin(t) + 1)^2 = 3 - 2(\cos(t) - \sin(t)), \\
    & \norm{(x,y) - (-1,-1)}_2^2 = (\cos(t) + 1)^2 + (\sin(t) + 1)^2 = 3 + 2(\cos(t) + \sin(t)),
\end{align*}
let $u(t) = \cos(t) + \sin(t)$ and let $w(t) = \cos(t) - \sin(t)$ and
$f(\alpha) = \sqrt{3 + 2\alpha} + \sqrt{3 - 2 \alpha}$. We get:
\[
    G(x,y) = G\prn*{\cos(t), \sin(t)} = \frac{1}{4} \prn*{f(u(t)) + f(w(t))}.
\]
We note that for $t \in [0, \frac{\pi}{4}]$,
$0\le w(t)\le u(t)\le \sqrt{2}$.

Let $H(t) := f(u(t)) + f(w(t))$.

\begin{align*}
    H'(t) &= f'(u(t)) \ u'(t) + f'(w(t)) \ w'(t) \\
     & = f'(u(t)) \ w(t) - f'(w(t)) \ u(t), \numberthis \label{eq:H-derivative}
\end{align*}
where the last equality follows from the fact that $u'(t) = -\sin(t) + \cos(t) = w(t)$ and $w'(t) = -\sin(t)-\cos(t) = -u(t)$.

\begin{clm}\label{clm:f-tag-tmp-eq}
    For any $0\le \alpha \le \beta \le \sqrt{2}$:
    \[
        f'(\beta)\alpha -f'(\alpha)\beta \le 0.
    \]
\end{clm}
\begin{proof}
Taking the derivative of $f$ we get:
$f'(\alpha) = \prn*{3+2\alpha}^{-\frac{1}{2}} - \prn*{3-2\alpha}^{-\frac{1}{2}}$. Let $q(\alpha) := - f'(\alpha)$. So:
\[
    q'(\alpha) = \prn*{3+2\alpha}^{-\frac{3}{2}} + \prn*{3 - 2 \alpha}^{-\frac{3}{2}},
\]
\[
    q''(\alpha) = - 3 \prn*{3 + 2\alpha}^{-\frac{5}{2}} + 3 \prn*{3 - 2\alpha}^{-\frac{5}{2}}.
\]

Let $l(x) = x^{-\frac{5}{2}}$. So $l'(x) = -\frac{5}{2} x^{\frac{-7}{2}} < 0$ which implies that $l$ is strictly decreasing on $x > 0$ and thus, since $3-2\alpha < 3 + 2\alpha$ for any $\alpha > 0$: 
$(3  - 2\alpha)^{-\frac{5}{2}} > (3 + 2\alpha)^{-\frac{5}{2}}$ and therefore $q''(\alpha) > 0$. % Similarly, $q'(\alpha) > 0$.

$q$ is thus convex and $q(0) = 0$. Thus for any $0 < \alpha < \beta$, by denoting $\alpha = \lambda \beta + (1-\lambda)\cdot 0$ where $\lambda = \frac{\alpha}{\beta} \in (0,1)$, we get  $q(\alpha) = q(\lambda \beta + (1-\lambda)0) \le \lambda q(\beta) + (1-\lambda)0 = \frac{\alpha}{\beta} q(\beta)$. Reorganizing yields: $q(\alpha)\beta \le q(\beta)\alpha$ or
\[
    f'(\beta)\alpha -f'(\alpha)\beta \le 0.
\]
The boundary cases $\alpha=0$ and $\alpha=\beta$ hold with equality.

\end{proof}
By plugging $\beta = u(t), \alpha = w(t)$ in Claim \ref{clm:f-tag-tmp-eq} we get (via \cref{eq:H-derivative}): $H'(t) \le 0$. Therefore $H(t)$ is non-increasing on $[0, \frac{\pi}{4}]$ which implies that 
\begin{align*}
    G(\cos(t),\sin(t)) &\le \frac{1}{4} H(0) \\
    & = \frac{1}{4} \prn*{\sqrt{3+2 u(0)} + \sqrt{3 - 2 u(0)} + \sqrt{3 + 2 w(0)} + \sqrt{3 - 2 w(0)}} \\
    & = \frac{1}{4} \prn*{2\sqrt{5} + 2} = \frac{1 + \sqrt{5}}{2} = \phi.
\end{align*}

\end{proof}

\section{Missing analysis from \cref{sec:lb-2-Rd}}\label{sec:missing-analysis-lb-2-Rd}

\subsection{Proof of \cref{lem:simplex-identities}}
\begin{proof}
First,
\begin{align*}
\left\lVert\sum_{i=0}^{d'}s_i\right\rVert_2^2
&=
\sum_{i=0}^{d'}\lVert s_i\rVert_2^2
+
2\sum_{0\le i<j\le d'}\langle s_i,s_j\rangle =
(d'+1)
-
\frac{2}{d'}\binom{d'+1}{2} = 0,
\end{align*}
which implies the first item. Next, fix $j\in\{0,\ldots,d'\}$.  Then
\begin{align*}
\sum_{i=0}^{d'}
\langle s_j,s_i\rangle s_i
&=
s_j-\frac{1}{d'}\sum_{i\ne j}s_i =
s_j+\frac{1}{d'}s_j =
\frac{d'+1}{d'}s_j,
\end{align*}
where the second equality follows from the first item (i.e. $\sum_{i=0}^{d'} s_i = 0$).
Since $s_0,\ldots,s_{d'}$ span $V$, it follows by linearity that, for
every $v\in V$, $\sum_{i=0}^{d'}
\langle v,s_i\rangle s_i
=
\frac{d'+1}{d'}v$.
Taking the inner product of both sides with $v$ gives
\[
\sum_{i=0}^{d'}
\langle v,s_i\rangle^2
=
\frac{d'+1}{d'}\lVert v\rVert_2^2.
\numberthis \label{eq:tmp1}
\]

Finally, let $w \in \R^d$. Since every $s_i\in V$ has  $\langle w,s_i\rangle=\langle w_V,s_i\rangle$, applying \cref{eq:tmp1} with $v=w_V$ and dividing by $d'+1$ proves the second item.
\end{proof}
\section{Missing analysis from \cref{sec:R2-lower-bounds}}\label{sec:missing-analysis-lb-R2}

\subsection{Proof of \cref{lem:equilateral-pointwise}}

The proof of \cref{lem:equilateral-pointwise} follows directly from the following two claims.
\begin{clm}\label{clm:y-avg-dist}
    For every $y\in\mathbb R^2$,
\[
\frac13
\sum_{j=0}^2
\norm{y-w_j}_2
\ge
g\left(\norm{y-c}_2\right).
\]
\end{clm}
\begin{proof}[Proof of \cref{clm:y-avg-dist}]
Let $r:=\norm{y-c}_2=\norm{y}_2$, and write
$y=r(\cos\theta,\sin\theta)$, choosing $\theta$ arbitrarily when
$r=0$. For $j\in\{0,1,2\}$, let $t_j := \cos\left(\theta-\frac{2\pi j}{3}\right)$, and let $
z_j:=\frac{2t_j+1}{3}$.
Since $t_j\in[-1,1]$, we have $z_j\in[-1/3,1]$.

We first calculate the first two moments of the $t_j$'s. Writing
$a:=\cos\theta$ and $b:=\sin\theta$, we have
$t_0=a$, $t_1=-\frac{a}{2}+\frac{\sqrt{3}b}{2}$, 
$t_2=-\frac{a}{2}-\frac{\sqrt{3}b}{2}$.
Therefore,
\begin{align*}
\sum_{j=0}^2t_j
&=
a
+
\left(-\frac{a}{2}+\frac{\sqrt{3}b}{2}\right)
+
\left(-\frac{a}{2}-\frac{\sqrt{3}b}{2}\right) = 0.
\end{align*}
And also:
\begin{align*}
\sum_{j=0}^2t_j^2
&=
a^2
+
\left(-\frac{a}{2}+\frac{\sqrt{3}b}{2}\right)^2
+
\left(-\frac{a}{2}-\frac{\sqrt{3}b}{2}\right)^2 \\
&=
a^2
+
\left(
\frac{a^2}{4}
-\frac{\sqrt{3}ab}{2}
+\frac{3b^2}{4}
\right)
+
\left(
\frac{a^2}{4}
+\frac{\sqrt{3}ab}{2}
+\frac{3b^2}{4}
\right) \\
&=
\frac{3}{2}a^2+\frac{3}{2}b^2 = \frac{3}{2},
\end{align*}
where the last equality follows from $a^2+b^2=1$.

It follows directly from $z_j=(2t_j+1)/3$ that
\begin{align*}
\sum_{j=0}^2z_j
&=
\sum_{j=0}^2\frac{2t_j+1}{3} =
\frac{2\sum_{j=0}^2t_j+3}{3} = 1.
\end{align*}
Similarly,
\begin{align*}
\sum_{j=0}^2z_j^2
&=
\sum_{j=0}^2\left(\frac{2t_j+1}{3}\right)^2 =
\frac{1}{9}
\sum_{j=0}^2
\left(4t_j^2+4t_j+1\right) \\
&=
\frac{
4\sum_{j=0}^2t_j^2
+
4\sum_{j=0}^2t_j
+
3
}{9} =
\frac{4\cdot\frac{3}{2}+4\cdot 0+3}{9} = 1.
\end{align*}
Let $s:=\sqrt{r^2+2r+4}$ and $
p:=\frac{\abs{r-2}}{s}$.
Since $s^2-\abs{r-2}^2=6r\ge 0$, we have $p\in[0,1]$ and
$s^2(1-p^2)=6r$.
The triangle vertex $w_j$ has norm $2$, and the angle between $y$ and $w_j$
is $\theta-2\pi j/3$. Note that $y=r(\cos\theta,\sin\theta)$, $w_j
=
2\left(
\cos\left(\frac{2\pi j}{3}\right),
\sin\left(\frac{2\pi j}{3}\right)
\right)$ and we have
\begin{align*}
\langle y,w_j\rangle
&=
2r\left(
\cos\theta\cos\left(\frac{2\pi j}{3}\right)
+
\sin\theta\sin\left(\frac{2\pi j}{3}\right)
\right) =
2r\cos\left(\theta-\frac{2\pi j}{3}\right) = 2rt_j.
\end{align*}
Thus, by the law of cosines,
\[
\norm{y-w_j}_2^2
=
\norm{y}_2^2+\norm{w_j}_2^2-2\langle y,w_j\rangle
=
r^2+4-4rt_j.
\]

On the other hand,
\begin{align*}
s^2\left(1-(1-p^2)z_j\right)
&=
s^2-6rz_j \\
&=
r^2+2r+4
-
6r\left(\frac{2t_j+1}{3}\right) \\
&=
r^2+2r+4-4rt_j-2r \\
&=
r^2+4-4rt_j.
\end{align*}
Consequently,
\[
\norm{y-w_j}_2
=
s\sqrt{1-(1-p^2)z_j}.
\]

We now use the following elementary inequality: for every
$p\in[0,1]$ and every $z\in[-1/3,1]$,
\[
\sqrt{1-(1-p^2)z}
\ge
1-\frac{1-p^2}{2}z-\frac{(1-p)^2}{2}z^2.
\]
Indeed, let $
q
:=
1-\frac{1-p^2}{2}z-\frac{(1-p)^2}{2}z^2$.
First, we have that
\[
q-p
=
\frac{(1-p)(1-z)\bigl(2+(1-p)z\bigr)}{2}
\ge 0,
\]
because $p\le 1$, $z\le 1$, and
\[
2+(1-p)z
\ge
2-\frac{1-p}{3}
>
0.
\]
Thus $q\ge p\ge 0$. A direct expansion also gives
\[
1-(1-p^2)z-q^2
=
\frac{(1-p)^3}{4}
z^2(1-z)\bigl(3+p+(1-p)z\bigr).
\]
Every factor on the right-hand side is nonnegative, where the non-negativity of the last term holds as
\[
3+p+(1-p)z
\ge
3+p-\frac{1-p}{3}
=
\frac{8+4p}{3}
>
0.
\]
Therefore,
\[
1-(1-p^2)z\ge q^2.
\]
Since $q\ge 0$, taking square root proves the claimed inequality.

Applying this inequality to each $z_j$ gives
\begin{align*}
\sum_{j=0}^2\norm{y-w_j}_2
&=
s\sum_{j=0}^2
\sqrt{1-(1-p^2)z_j} \\
&\ge
s\sum_{j=0}^2
\left(
1-\frac{1-p^2}{2}z_j
-\frac{(1-p)^2}{2}z_j^2
\right) \\
&=
s\left(
3
-\frac{1-p^2}{2}\sum_{j=0}^2z_j
-\frac{(1-p)^2}{2}\sum_{j=0}^2z_j^2
\right) \\
&=
s\left(
3-\frac{1-p^2}{2}-\frac{(1-p)^2}{2}
\right) \\
&=
s(2+p) \\
&=
2\sqrt{r^2+2r+4}+\abs{r-2}.
\end{align*}
Dividing by $3$ and recalling that $r=\norm{y-c}_2$, we obtain
\[
\frac{1}{3}\sum_{j=0}^2\norm{y-w_j}_2
\ge
\frac{
\abs{r-2}
+
2\sqrt{r^2+2r+4}
}{3}
=
g\left(\norm{y-c}_2\right).
\]
\end{proof}

We show the properties of $g$ in the following claim.
\begin{clm}\label{clm:g-convex-nondec}
    $g$ is convex and nondecreasing on $[0, \infty)$.
\end{clm}
\begin{proof}[proof of \cref{clm:g-convex-nondec}]
We rewrite $g(r)$ as $g(r) = \frac{1}{3}\abs*{r-2} + \frac{2}{3}\sqrt{(r+1)^2+3}$.
The first term $\frac{1}{3}\abs*{r-2}$ is convex. The second term is proportional to the composition of the Euclidean norm $\norm*{\cdot}_2$ with the affine map $r \mapsto (r+1, \sqrt{3})$. Since the norm is convex and affine compositions preserve convexity, the second term is convex. Being a sum of convex functions, $g(r)$ is convex on $[0, \infty)$.

To show $g$ is nondecreasing, we check its derivative. For $r > 2$, 
\[
g'(r) = \frac{1}{3}\prn*{1 + \frac{2(r+1)}{\sqrt{r^2+2r+4}}} > 0.
\]
For $0 < r < 2$, 
\[
g'(r) = \frac{1}{3}\prn*{-1 + \frac{2(r+1)}{\sqrt{r^2+2r+4}}}.
\]
For $g'(r) \ge 0$, we require $2(r+1) \ge \sqrt{r^2+2r+4}$. Since $r \ge 0$, both sides are positive, and squaring yields:
\[
4(r^2+2r+1) \ge r^2+2r+4 \iff 3r^2+6r \ge 0,
\]
which trivially holds for all $r \ge 0$. Because $g(r)$ is continuous and its left and right derivatives at $r=2$ are non-negative, $g(r)$ is nondecreasing on $[0, \infty)$.
\end{proof}

\subsection{Proof of \cref{lem:circle-average-and-discretization}}

\begin{clm}
\label{clm:circle-average-properties}
For every $D>0$, the function $\Gamma_D$ is convex and nondecreasing.
\end{clm}
\begin{proof}
For $\theta\in[0,2\pi]$, let $ f_\theta(r) := \sqrt{r^2+D^2-2rD\cos\theta} = \left\|(r,0)-D(\cos\theta,\sin\theta)\right\|_2$.
Since the Euclidean norm is convex and
$r\mapsto (r,0)-D(\cos\theta,\sin\theta)$ is affine, $f_\theta$ is
convex as a function of $r\in\R$. Hence their average $\Gamma_D(r) = \frac{1}{2\pi}\int_0^{2\pi} f_\theta(r)\,d\theta$ is convex. It remains to show monotonicity. Observe that $\Gamma_D$ is even:
\begin{align*}
\Gamma_D(-r) &= \frac{1}{2\pi}\int_0^{2\pi} \sqrt{r^2+D^2+2rD\cos\theta}\,d\theta =
\frac{1}{2\pi}\int_0^{2\pi}
\sqrt{r^2+D^2-2rD\cos(\theta+\pi)}\,d\theta =
\Gamma_D(r).
\end{align*}
We use this to show that $\Gamma_D$ is nondecreasing on $\R_{\ge 0}$.
for any  $0\le r\le s$: $r = \frac{s+r}{2s}s+\frac{s-r}{2s}(-s)$, and therefore, by the fact that $\Gamma_D$ is convex and even,
\[ \Gamma_D(r) \le \frac{s+r}{2s}\Gamma_D(s) + \frac{s-r}{2s}\Gamma_D(-s) = \Gamma_D(s). \]
\end{proof}

\begin{clm}
\label{clm:circle-discretization-error}
Let $b,y\in\mathbb R^2$, let $D>0$, and for any $j \in \crl*{0,\ldots,k-1}$ let\\$z_j = b+ D \prn*{ \cos\prn*{\frac{2\pi j}{k}}, \sin\prn*{\frac{2\pi j}{k}}}$.
Then
\[
\frac1k
\sum_{j=0}^{k-1}
\norm{y-z_j}_2
\ge
\Gamma_D(\norm{y-b}_2)
-
\frac{\pi D}{2k}.
\]
\end{clm}
\begin{proof}
Translate the plane so that $b=0$, and then rotate it so that
$y=(r,0)$, where $r=\norm{y-b}_2$. This rotation also rotates the
discretization grid. Consequently, for some phase $\delta\in[0,2\pi)$,
the rotated grid points have angles
$\theta_j=\delta+2\pi j/k$, for $j=0,\ldots,k-1$.
Let $f(\theta) := \norm{y-D(\cos\theta,\sin\theta)}_2$. Then
$\Gamma_D(r) = \frac{1}{2\pi}\int_0^{2\pi} f(\theta)\,d\theta$ and
$\frac{1}{k}\sum_{j=0}^{k-1}\norm{y-z_j}_2
= \frac{1}{k}\sum_{j=0}^{k-1}f(\theta_j)$.
We first show the following claim.
% \noindent
\begin{clm}\label{clm:f-is-lipschitz}
The function $f$ is $D$-Lipschitz; that is, for every
$\theta,\phi\in\mathbb R$,
\[
|f(\theta)-f(\phi)|
\le
D|\theta-\phi|.
\]
\end{clm}
\begin{proof}[Proof of \cref{clm:f-is-lipschitz}]
By the reverse triangle inequality,
\begin{align*}
|f(\theta)-f(\phi)|
&=
\left|
\norm{y-D(\cos\theta,\sin\theta)}_2
-
\norm{y-D(\cos\phi,\sin\phi)}_2
\right| \\
&\le
D\norm{(\cos\theta,\sin\theta)-(\cos\phi,\sin\phi)}_2 \\
&=
D\sqrt{
(\cos\theta-\cos\phi)^2
+
(\sin\theta-\sin\phi)^2
} \\
&= D\sqrt{\cos^{2}(\theta) + \sin^2(\theta) + \cos^2(\phi) + \sin^2(\phi) - 2 \prn*{\cos(\theta)\cos(\phi) + \sin(\theta)\sin(\phi)}} \\
&=
D\sqrt{2-2\cos(\theta-\phi)} = 2D\left| \sin\prn*{\frac{\theta-\phi}{2}} \right| \le D|\theta-\phi|,
\end{align*}
where the last inequality follows from $|\sin x|\le |x|$.
This proves the claim.
\end{proof}

We now partition the circle into the $k$ arcs centered at the phase-shifted points $\theta_j$. Namely, let $I_j = \brk*{\theta_j-\frac{\pi}{k}, \theta_j+\frac{\pi}{k} }$,
where angles are understood modulo $2\pi$. Each arc has length
$2\pi/k$, and the arcs $I_0,\ldots,I_{k-1}$ partition the circle.

Fix $j$. For every $\theta\in I_j$, \cref{clm:f-is-lipschitz} gives
$f(\theta)\le f(\theta_j)+D|\theta-\theta_j|$. Therefore,
\begin{align*}
\int_{I_j} f(\theta)\,d\theta
&\le
\int_{I_j}
\left(
f(\theta_j)+D|\theta-\theta_j|
\right)\,d\theta =
\frac{2\pi}{k}f(\theta_j)
+
D\int_{-\pi/k}^{\pi/k}|t|\,dt =
\frac{2\pi}{k}f(\theta_j)
+
\frac{\pi^2D}{k^2}.
\end{align*}
Summing over $j=0,\ldots,k-1$ yields
\[
\int_0^{2\pi}f(\theta)\,d\theta
\le
\frac{2\pi}{k}
\sum_{j=0}^{k-1}f(\theta_j)
+
\frac{\pi^2D}{k}.
\]
Dividing by $2\pi$, we obtain
\[
\Gamma_D(r)
\le
\frac{1}{k}
\sum_{j=0}^{k-1}f(\theta_j)
+
\frac{\pi D}{2k}.
\]
Since $r=\norm{y-b}_2$ and $f(\theta_j)=\norm{y-z_j}_2$, by rearranging we get
\[ \frac{1}{k} \sum_{j=0}^{k-1}\norm{y-z_j}_2 \ge \Gamma_D(\norm{y-b}_2) - \frac{\pi D}{2k}. \]
\end{proof}

\subsection{Proof of \cref{lem:farthest-vertex-and-radial-convexity}}

\begin{clm}
\label{clm:farthest-discretized-circle-vertex}
For every $y\in\mathbb R^2$, and $r:=\norm{y-b}_2$ we have
\[
\max_{0\le j\le k-1}
\norm{y-z_j}_2
\ge
\sqrt{
r^2+D^2+
2rD\cos\prn*{\frac{\pi}{k}}
}.
\]
\end{clm}
\begin{proof}
If $r=0$, both sides equal $D$, so suppose that $r>0$.

Among the $k$ equally spaced directions of the vectors
$
z_j-b,
$
there exists one whose angular distance from the direction
$
-\frac{y-b}{\norm{y-b}_2}
$
is at most $\pi/k$. Let $z_j$ be such a vertex, and let
$
\varepsilon\le\frac{\pi}{k}
$
be this angular distance.

The angle between $y-b$ and $z_j-b$ is $\pi-\varepsilon$.
Therefore, by the law of cosines,
\[
\begin{aligned}
\norm{y-z_j}_2^2
&=
r^2+D^2
-
2rD\cos(\pi-\varepsilon)\\
&=
r^2+D^2+
2rD\cos\varepsilon\\
&\ge
r^2+D^2+
2rD\cos\prn*{\frac{\pi}{k}}.
\end{aligned}
\]
Taking square roots proves the claim.
\end{proof}

\begin{clm}
\label{clm:ex-post-radial-convexity}
$H_{D,k}$ is convex and nondecreasing on $[0,\infty)$.
\end{clm}
\begin{proof}
Set
$
\kappa:=\cos\prn*{\frac{\pi}{k}}\ge0.
$
Then
$
H_{D,k}(r)
=
\sqrt{
(r+D\kappa)^2+
D^2(1-\kappa^2)
}.
$
Its first derivative is
$
H_{D,k}'(r)
=
\frac{r+D\kappa}
{H_{D,k}(r)}
\ge0,
$
and its second derivative is
$
H_{D,k}''(r)
=
\frac{
D^2(1-\kappa^2)
}{
H_{D,k}(r)^3
}
\ge0.
$
Thus, $H_{D,k}$ is nondecreasing and convex.
\end{proof}

\section{A two-agent $\sqrt{2} \approx 1.41$ upper bound for both objectives in $\R^2$}\label{sec:2-agents-upper-bound}

In the one-dimensional case, the best strategyproof in expectation mechanism for $2$ agents is the LRM mechanism of \citet{procaccia2013approximate}, which mixes the random dictator mechanism with the optimal (non-strategyproof) solution; that is, it returns the leftmost/rightmost point w.p. $0.25$ each or the optimal solution w.p. $0.5$.
This mechanism yields a tight expected approximation ratio of $1.5$. Surprisingly, \citet{balkanski2024randomized} note that the lower bound of $1.5$ for $2$ agents in $\R$ does not generalize to $\R^2$. In fact, they show that the lower bound for $2$ agents in $\R^2$ is only approximately $1.118$.
This raises the question of whether there is a better strategyproof in expectation mechanism for $2$ agents in $\R^2$.

Surprisingly, in two dimensions, even though every two points lie on a single line, there's a better ``outside the box'' solution: return a point which is not on the line connecting them at all, with a better expected approximation ratio of $\sqrt{2} \approx 1.41$.

\begin{thm}\label{thm:RIRT}
    Consider an instance with $2$ agents in $\R^2$.
    $\rirt$ is strategyproof in expectation and has an expected approximation ratio of at most $\sqrt{2}$ for both the ex-post and ex-ante objectives.
\end{thm}

\paragraph{The high level idea.}
The idea: we construct a random isosceles right triangle with $AB$ as its hypotenuse. Given agents at $A,B \in \R^2$, our mechanism $M$ returns $F = M(A,B)$ such that $AFB$ is an isosceles right triangle. There are two options to choose such a right triangle for each pair of agent locations $A$ and $B$: choosing the point below the segment $AB$ and the one above it.
The formula for each option is $F_{\pm}(A,B) = \frac{A+B}{2} \pm R \prn*{\frac{B-A}{2}}$, where $R \in M_{2 \times 2}(\R)$ is the $\frac{\pi}{2}$ rotation matrix \[ R := \begin{pmatrix} 0 & -1 \\ 1 & 0 \end{pmatrix}. \]
As we show, neither of the two fixed deterministic branch rules is strategyproof.
\begin{clm}\label{clm:deterministic-triangle-not-sp}
For each $\sigma\in\{-1,+1\}$, the deterministic mechanism
\[
F_\sigma(A,B):=\frac{A+B}{2}+\sigma R\prn*{\frac{B-A}{2}}
\]
is not strategyproof.
\end{clm}
\begin{proof}
Fix $\sigma\in\{-1,+1\}$ and consider agents at $A=(0,0)$ and
$B=(1,0)$. Under truthful reporting,
$F_\sigma(A,B)=(1/2,\sigma/2)$, so agent $B$ incurs cost
$1/\sqrt{2}$. If agent $B$ instead reports $B'=(1,-\sigma)$, then
\[
F_\sigma(A,B')
=\frac{(1,-\sigma)}{2}
+\sigma R\prn*{\frac{(1,-\sigma)}{2}}
=(1,0)=B.
\]
The deviation reduces the agent's cost to zero. Thus neither fixed
branch mechanism is strategyproof.
\end{proof}

The good news is that randomizing over the two options (choosing each one w.p. $0.5$) is strategyproof.
The mechanism thus returns $F_+(A,B)$ w.p. $0.5$ and $F_-(A,B)$ w.p. $0.5$. We formally define the mechanism in \cref{alg:rirt}.

\begin{algorithm}[H]
\caption{$\rirt$ Mechanism}
\label{alg:rirt}
\begin{algorithmic}[1]
\REQUIRE Reported locations $A,B \in \R^2$
\STATE Sample $s \sim \mathrm{Unif}(\crl*{+,-})$.
\STATE Let $R \in M_{2 \times 2}(\R)$ be the $\frac{\pi}{2}$ rotation matrix.
\STATE Let $F_{+}(A,B) = \frac{A+B}{2} + R \prn*{\frac{B-A}{2}}$,\ \  $F_{-}(A,B) = \frac{A+B}{2} - R \prn*{\frac{B-A}{2}}$.
\STATE Return $F_s(A,B)$.
\end{algorithmic}
\end{algorithm}

In the case of the real line the optimal mechanism returns the optimal solution with probability half, and makes sure that no agent may profitably deviate by returning the leftmost/rightmost point w.p. $0.25$ each. In our case, we return an off-line point whose orthogonal projection onto the line connecting the agents is the optimal midpoint. The intuition is that if an agent deviates it might help under one branch, but under the other branch it hurts the agent by an offsetting amount because of the symmetry of the construction.

We first show that $\rirt$ is strategyproof in expectation, and then we show that it has an expected approximation ratio of at most $\sqrt{2}$.
\begin{thm}\label{thm:RIRT-SP}
    Consider an instance with $2$ agents in $\R^2$.
    $\rirt$ is strategyproof in expectation.
\end{thm}
\begin{proof}
Assume w.l.o.g. that agent $1$ is at the origin, and agent $2$ is at some point $v \in \R^2$ (otherwise we may translate the instance). Consider a deviation of agent $1$ from the truthful report $0$ to some point $x \in \R^2$. Let $C(x)$ be the expected cost of agent $1$ when reporting $x$.

By the definition of the matrix $R$, it satisfies $R^2=-I$. Expanding the two possible outcomes gives
\begin{align*}
F_+(x,v)+R F_-(x,v)
&= \frac{x+v}{2}+R\prn*{\frac{v-x}{2}}
   +R\prn*{\frac{x+v}{2}}-R^2\prn*{\frac{v-x}{2}} \\
&=v+Rv.
\end{align*}
Therefore, since the rotation $R$ preserves Euclidean norms, the triangle inequality gives
\begin{align*}
    C(x)
    &=\frac{1}{2}\prn*{\norm{F_+(x,v)}_2+\norm{F_-(x,v)}_2} \\
    &=\frac{1}{2}\prn*{\norm{F_+(x,v)}_2+\norm{R F_-(x,v)}_2} \\
    &\ge \frac{1}{2}\norm{F_+(x,v)+R F_-(x,v)}_2 \\
    &=\frac{1}{2}\norm{v+Rv}_2
     =\frac{\norm{v}_2}{\sqrt{2}},
\end{align*}
where the last equality follows because $v$ and $Rv$ are orthogonal and have the same norm. On the other hand, under truthful reporting,
\begin{align*}
C(0)
&=\frac{1}{2}\prn*{
    \norm{\frac{v+Rv}{2}}_2+
    \norm{\frac{v-Rv}{2}}_2}
=\frac{\norm{v}_2}{\sqrt{2}}.
\end{align*}
Thus $C(x)\ge C(0)$ for every report $x$, so agent $1$ cannot benefit in expectation from a deviation. Swapping the two reports merely interchanges $F_+$ and $F_-$, so the same argument applies to agent $2$. Hence, $\rirt$ is strategyproof in expectation.

\end{proof}

Next, we show that this mechanism has an expected approximation ratio of at most $\sqrt{2}$.
\begin{thm}\label{thm:RIRT-AR}
    Consider an instance with $2$ agents in $\R^2$.
    $\rirt$ has an expected approximation ratio of at most $\sqrt{2}$ for both the ex-post and ex-ante objectives.
\end{thm}
\begin{proof}
Since the ex-post objective is always at least the ex-ante objective, it suffices to show that the expected approximation ratio of $\rirt$ is at most $\sqrt{2}$ for the ex-post objective.
Let the two reported locations be $A,B\in\R^2$, and let $D := \norm{A-B}_2$.
Denote by $M := \frac{A+B}{2}$ the midpoint.

\textbf{The optimal egalitarian cost:}
For any facility location $z\in\R^2$, by the triangle inequality, we have
\[
    \norm{A-B}_2 \le \norm{A-z}_2 + \norm{z-B}_2 \le 2\max\{\norm{A-z}_2,\norm{B-z}_2\}.
\]
Hence $\max\{\norm{A-z}_2,\norm{B-z}_2\} \ge D/2$ for every $z$, and equality is achieved by $z=M$.
Therefore the optimal egalitarian cost is
\begin{equation}\label{eq:OPT-2-agents}
    \OPT(A,B) = \min_{z\in\R^2}\max\{\norm{A-z}_2,\norm{B-z}_2\} = \frac{D}{2}.
\end{equation}

\textbf{The cost of $\rirt$:}
The mechanism outputs $F_+(A,B)$ w.p.\ $1/2$ and $F_-(A,B)$ w.p.\ $1/2$, where
\[
    F_{\pm}(A,B) = \frac{A+B}{2} \pm R\!\left(\frac{B-A}{2}\right).
\]
By construction, $\triangle ABF_\pm$ is an isosceles right triangle with hypotenuse $AB$.
Thus:
\[
    \norm{F_\pm(A,B)-A}_2 = \norm{F_\pm(A,B)-B}_2 = \frac{\norm{A-B}_2}{\sqrt{2}} = \frac{D}{\sqrt{2}}.
\]
Therefore, for either realized outcome $F_\pm(A,B)$, the egalitarian cost equals
\[
    \max\{\norm{A-F_\pm(A,B)}_2,\norm{B-F_\pm(A,B)}_2\}=\frac{D}{\sqrt{2}}.
\]
Therefore the expected egalitarian cost of the mechanism is also $D/\sqrt{2}$.

By combining this with \cref{eq:OPT-2-agents}, we get
\[
    \frac{\E\big[\max\{\norm{A-Z}_2,\norm{B-Z}_2\}\big]}{\OPT(A,B)}
    = \frac{D/\sqrt{2}}{D/2}
    = \sqrt{2}.
\]
Hence \rirt has an expected approximation ratio at most $\sqrt{2}$.
\end{proof}

\printbibliography

\end{document}